\documentclass[a4paper,UKenglish]{lipics-v2016}

\usepackage{microtype,xspace,stmaryrd}

\newcommand{\NP}{{\sf NP}\xspace}

\newcommand{\Acal}{\mathcal{A}}

\newcommand{\W}{{\sf W}\xspace}

\theoremstyle{plain}

\newtheorem{proposition}[theorem]{Proposition}

\newenvironment{proof-claim}[1][]{\par \noindent {\textcolor{darkgray}{\sffamily\bfseries Proof of the claim.#1}}\ }{\hfill$\Box$\bigskip}

\newenvironment{proof-sketch}[1][]{\par \noindent {\textcolor{darkgray}{\sffamily\bfseries Sketch of proof.#1}}\ }{\hfill$\Box$\bigskip}

\newenvironment{proof-DAG}[1][]{\par \noindent {\textcolor{darkgray}{\sffamily\bfseries Proof of Theorem~\ref{thm:algoDAG}.#1}}\ }{\hfill$\Box$\bigskip}

\usepackage{cite}
\usepackage{hyperref}
\usepackage{boxedminipage}

\newcommand{\np}{{\sf NP}\xspace}

\newcommand{\fpt}{{\sf FPT}\xspace}
\newcommand{\xp}{{\sf XP}\xspace}

\newcommand{\igjournal}[1]{}

\newtheorem{claim}{Claim}

\title{On the complexity of finding internally vertex-disjoint long directed paths\footnote{Emails of authors: \texttt{julio@mat.ufc.br}, \texttt{campos@lia.ufc.br}, \texttt{karolmaia@lia.ufc.br}, \texttt{ignasi.sau@lirmm.fr}, \texttt{anasilva@mat.ufc.br}. \newline Work supported by DE-MO-GRAPH grant ANR-16-CE40-0028 and CNPq grant 306262/2014-2.}}

\titlerunning{On the complexity of finding internally vertex-disjoint long directed paths}

\author[1]{J\'ulio Ara\'ujo}
\author[2]{Victor A. Campos}
\author[2]{Ana Karolinna Maia}
\author[1,3]{Ignasi Sau}
\author[1]{Ana Silva}

\affil[1]{Departamento de Matemática, ParGO research group, Universidade Federal do Ceará, Fortaleza, Brazil}

\affil[2]{Departamento de Computa\c{c}\~{a}o, ParGO research group, Universidade Federal do Ceará, Fortaleza, Brazil}

\affil[3]{CNRS, AlGCo project team, LIRMM, Montpellier, France}

\authorrunning{J. Ara\'ujo, V. A. Campos, A. K. Maia, I. Sau, and A. Silva} 

\Copyright{J\'ulio Ara\'ujo, Victor A. Campos, Ana Karolinna Maia, Ignasi Sau, and Ana Silva}


\EventEditors{John Q. Open and Joan R. Acces}
\EventNoEds{2}
\EventLongTitle{}
\EventShortTitle{}
\EventAcronym{}
\EventShortTitle{CVIT 2017}
\EventAcronym{}
\EventDate{4--8 September 2017}
\EventLocation{Vienna, Austria}
\EventLogo{}

\begin{document}

\maketitle


\begin{abstract}
For two positive integers $k$ and $\ell$, a \emph{$(k \times \ell)$-spindle} is the union of $k$ pairwise internally vertex-disjoint directed paths with $\ell$ arcs between two vertices $u$ and $v$. We are interested in the  (parameterized) complexity of
several problems consisting in deciding whether a given digraph contains a
subdivision of a spindle, which generalize both the
\textsc{Maximum Flow} and \textsc{Longest Path} problems. We obtain the following complexity dichotomy: for a fixed $\ell \geq 1$, finding the largest $k$ such that an input digraph $G$ contains a subdivision of a $(k \times \ell)$-spindle is polynomial-time solvable if $\ell \leq 3$, and \NP-hard otherwise. We place special emphasis on finding spindles with exactly two paths and present {\sf FPT} algorithms that are asymptotically optimal under the {\sf ETH}. These algorithms are based on the technique of representative families in matroids, and use also color-coding as a subroutine. Finally, we study the case where the input graph is acyclic, and present several algorithmic and hardness results.

\end{abstract}

\section{Introduction}
\label{sec:intro}

A \emph{subdivision} of a digraph $F$
is a digraph obtained from $F$ by replacing each arc $(u,v)$ of $F$ by a directed $(u,v )$-path. We are interested in the (parameterized) complexity of several problems consisting in deciding whether a given digraph contains as a subdigraph a subdivision of a \emph{spindle}, defined as follows.
For $k$ positive integers $\ell_1, \ldots, \ell_k$, a $(\ell_1, \ldots, \ell_k)$-\emph{spindle} is the digraph containing $k$ paths $P_1, \ldots, P_k$ from a vertex $u$ to a vertex $v$, such that $|E(P_i)| = \ell_i$ for $1 \leq i \leq k$ and $V(P_i) \cap V(P_j) = \{u,v\}$ for $1 \leq i \neq j \leq k$. If $\ell_i = \ell$ for $1 \leq i \leq k$, a $(\ell_1, \ldots, \ell_k)$-spindle is also called a  $(k \times \ell)$-\emph{spindle}. See Figure~\ref{fig:example} for an example.

\begin{figure}[h!]
\begin{center}\vspace{-.1cm}
\includegraphics[width=0.35\textwidth]{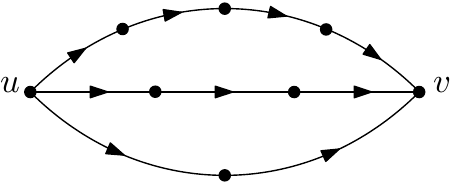}
\end{center}\vspace{-.15cm}
\caption{A $(4,3,2)$-spindle. This digraph contains a subdivision of a $(3 \times 2)$-spindle, but not of a $(3 \times 3)$-spindle.}
\label{fig:example}
\end{figure}

 Note that a digraph $G$ contains a subdivision of a $(k,1)$-spindle if and only if there exist two vertices $u$ and $v$ and $k$ internally vertex-disjoint paths from $u$ to $v$. On the other hand, $G$ contains a subdivision of a $(1, \ell)$-spindle if and only if $G$ contains a path of length at least $\ell$.
 Hence, finding a subdivision of a spindle generalizes both the
\textsc{Maximum Flow} and \textsc{Longest Path} problems.

Subdivisions of spindles were considered by Bang-Jensen et al.~\cite{BHM15}, who introduced the general problem of finding a subdivision of a {\sl fixed} digraph $F$ and presented {\sf NP}-hardness results and polynomial-time algorithms for several choices of $F$. In particular, they proved that when $F$ is a spindle, the problem can be solved in time $n^{O(|V(F)|)}$ by a simple combination of brute force and a flow algorithm. Using terminology from parameterized complexity, this means that the problem is in {\sf XP} parameterized by the size of $F$, and they left open whether it is {\sf FPT}. Note that on undirected graphs the notion of subdivision coincides with that of topological minor, hence by the results of Grohe et al.~\cite{GroheKMW11} the problem is {\sf FPT} parameterized by the size of $F$, for a general digraph $F$. We refer to the introduction of~\cite{BHM15} for a more detailed discussion about problems related to containment relations on graphs and digraphs.

\medskip

We first consider the following two optimization problems about finding subdivisions of spindles:
\begin{itemize}
\item[(1)] for a fixed positive integer $k$, given an input digraph $G$, find the largest integer $\ell$ such that $G$ contains a subdivision of a $(k \times \ell)$-spindle, and
\item[(2)] for a fixed positive integer $\ell$, given an input digraph $G$, find the largest integer $k$ such that $G$ contains a subdivision of a $(k \times \ell)$-spindle.
\end{itemize}
We call these problems \textsc{Max $(k \times \bullet)$-Spindle Subdivision} and \textsc{Max $(\bullet \times \ell)$-Spindle Subdivision}, respectively.
We prove that the first problem is \NP-hard for any integer $k \geq 1$, by a simple reduction from \textsc{Longest Path}. The second problem turns out to be much more interesting, and we achieve the following dichotomy.

\begin{theorem}\label{thm:dichotomy}
Let $\ell \geq 1$ be a fixed integer. \textsc{Max $(\bullet \times \ell)$-Spindle Subdivision} is polynomial-time solvable if $\ell \leq 3$, and \NP-hard otherwise, even restricted to acyclic digraphs.
\end{theorem}

The reduction for the \NP-hard cases is inspired by a result of Brewster et al.~\cite{BrewsterHPRY03} to prove the \NP-hardness of packing vertex-disjoint paths on digraphs. Concerning the polynomial algorithms, to solve the case $\ell = 3$, which is the only non-trivial one, we use a  \emph{vertex splitting procedure} that builds on ideas of Schrijver~\cite{Schrijver01} on undirected graphs and by Kriesell~\cite{Kri05} on directed graphs (see also~\cite[Section 5.9]{BJ-Gutin-book}).

It it worth mentioning that both the positive and negative results of Theorem~\ref{thm:dichotomy} hold as well for the case where the endvertices of the desired spindle are {\sl fixed}. Itai et al.~\cite{IPS82} considered the problems of,  given a digraph $G$ and two distinct vertices $s$ and $t$, finding the maximum number of internally vertex-disjoint $(s,t)$-paths whose lengths are {\sl at most} or {\sl exactly equal} to a fixed constant $\ell$, and achieved dichotomies for both cases. Note that the problem we consider corresponds to a constraint of type `{\sl at least}' on the lengths of the desired paths. Hence, Theorem~\ref{thm:dichotomy} together with the results of Itai et al.~\cite{IPS82}  provide a full picture of the complexity of finding a maximum number of length-constrained internally vertex-disjoint directed $(s,t)$-paths.

\medskip

We place special emphasis on finding subdivisions of spindles with exactly two paths, which we call \emph{$2$-spindles}. The existence of subdivisions of 2-spindles has attracted some interest in the literature. Indeed, Benhocine and Wojda~\cite{BenhocineW83} showed that a tournament on $n \geq 7$ vertices always contains a subdivision of a $(\ell_1, \ell_2)$-spindle such that $\ell_1 + \ell_2 = n$.  More recently, Cohen et al.~\cite{two-blocks-1} showed that a strongly connected digraph with chromatic number $\Omega((\ell_1 + \ell_2)^4)$ contains a subdivision of a $(\ell_1, \ell_2)$-spindle, and this bound was subsequently improved to $\Omega((\ell_1 + \ell_2)^2)$ by Kim et al.~\cite{two-blocks-2}, who also provided improved bounds for Hamiltonian digraphs.
%

We consider two  problems concerning the existence of subdivisions of 2-spindles. The first one is, given an input digraph $G$, find the largest integer $\ell$ such that $G$ contains a subdivision of a $(\ell_1,\ell_2)$-spindle with $\min\{\ell_1,\ell_2\} \geq 1$ and $\ell_1 + \ell_2 = \ell$. We call this problem \textsc{Max $(\bullet,\bullet)$-Spindle Subdivision}, and we show the following results.

\begin{theorem}\label{thm:2blockFPTsum}
Given a digraph $G$ and a positive integer $\ell$, the problem of deciding whether there exist two strictly positive integers $\ell_1,\ell_2$ with $\ell_1+\ell_2 = \ell$ such that $G$ contains a subdivision of a $(\ell_1,\ell_2)$-spindle is {\sf NP}-hard and {\sf FPT} parameterized by $\ell$. The running time of the  {\sf FPT} algorithm is $2^{O(\ell)}\cdot n^{O(1)}$, which is asymptotically optimal unless the {\sf ETH}   fails. Moreover, the problem does not admit polynomial kernels unless ${\sf NP} \subseteq {\sf coNP} / {\sf poly}$.
\end{theorem}

The second problem is, for a fixed strictly positive integer $\ell_1$, given an input digraph $G$, find the largest integer $\ell_2$ such that $G$ contains a subdivision of a $(\ell_1,\ell_2)$-spindle. We call this problem \textsc{Max $(\ell_1,\bullet)$-Spindle Subdivision}, and we show the following results.

\begin{theorem}\label{thm:2blockFPT-2length}
Given a digraph $G$ and two integers $\ell_1, \ell_2$ with $\ell_2 \geq \ell_1 \geq 1$, the problem of deciding whether $G$ contains a subdivision of a $(\ell_1,\ell_2)$-spindle can be solved in time $2^{O(\ell_2)}\cdot n^{O(\ell_1)}$. When $\ell_1$ is a constant, the problem remains {\sf NP}-hard and the running time of the {\sf FPT} algorithm parameterized by $\ell_2$ is asymptotically optimal unless the {\sf ETH}  fails. Moreover, the problem does not admit polynomial kernels unless ${\sf NP} \subseteq {\sf coNP} / {\sf poly}$.
\end{theorem}

The hardness results of Theorems~\ref{thm:2blockFPTsum} and~\ref{thm:2blockFPT-2length} are based on a simple reduction from \textsc{Directed Hamiltonian Cycle}. Both {\sf FPT} algorithms, which are our main technical contribution,  are based on the technique of \emph{representative families} in matroids introduced by Monien~\cite{Monien85}, and in particular its improved version recently presented by
Fomin et al.~\cite{FominLPS16}. The {\sf FPT} algorithm of Theorem~\ref{thm:2blockFPT-2length} also uses the \emph{color-coding} technique of Alon et al.~\cite{AlonYZ95} as a subroutine.

%
%
%
%
%
%
%
%
%

\medskip

Finally, we consider the case where the input digraph $G$ is {\sl acyclic}. We prove the following result by using a standard dynamic programming algorithm.

\begin{theorem}\label{thm:algoDAG}
Given an acyclic digraph $G$ and two positive integers $k,\ell$, the problem of deciding whether $G$ contains a subdivision of a $(k \times \ell)$-spindle can be solved in time $O(\ell^k \cdot n^{2k+1})$.
\end{theorem}

The above theorem implies, in particular, that when $k$ is a constant the problem is polynomial-time solvable on acyclic digraphs, which generalizes the fact that \textsc{Longest Path}, which corresponds to the case $k=1$, is polynomial-time solvable on acyclic digraphs (cf.~\cite{Book-algo}).

As observed by Bang-Jensen et al.~\cite{BHM15}, from the fact that the $k$-\textsc{Linkage} problem is in \xp on acyclic digraphs~\cite{Met93}, it easily follows that finding a subdivision of a general digraph $F$ is in \xp on DAGs parameterized by $|V(F)|$. Motivated by this, we prove two further hardness results about finding subdivisions of spindles on DAGs. Namely, we prove that
if $F$ is the disjoint union of $(2 \times 1)$-spindles, then finding a subdivision of $F$ is {\sf NP}-complete on planar DAGs, and that if $F$ is  the disjoint union of a $(k_1 \times 1)$-spindle and a $(k_2 \times 1)$-spindle, then finding a subdivision of $F$ is {\sf W}$[1]$-hard on DAGs parameterized by $k_1 + k_2$. These two results should be compared to the fact that finding a subdivision of a single $(k \times 1)$-spindle can be solved in polynomial time on general digraphs by a flow algorithm.

\medskip
\noindent \textbf{Organization of the paper}. In Section~\ref{sec:prelim} we provide some definitions about (di)graphs, parameterized complexity, and matroids.
In Section~\ref{sec:dichomoty} we prove Theorem~\ref{thm:dichotomy}, and
in Section~\ref{sec:two-paths} we prove Theorem~\ref{thm:2blockFPTsum} and Theorem~\ref{thm:2blockFPT-2length}. In Section~\ref{sec:DAGs} we focus on acyclic digraphs and we prove, in particular, Theorem~\ref{thm:algoDAG}. In Section~\ref{sec:further} we present some open problems for further research.


\section{Preliminaries}
\label{sec:prelim}

\noindent \textbf{Graphs and digraphs}. We use standard graph-theoretic notation, and we refer the reader to the books~\cite{Diestel12} and~\cite{BJ-Gutin-book} for any undefined notation about graphs and directed graphs, respectively.

A \emph{directed graph} $G$, or just \emph{digraph},  consists of a non-empty set $V(G)$ of elements called  \emph{vertices} and a finite (multi)set $A(G)$ of ordered pairs of distinct vertices called \emph{arcs}. All our positive results hold even for digraphs where multiple arcs between the same pair of vertices are allowed. We denote by $(u,v)$ an arc from a vertex $u$ to a vertex $v$. Vertex $u$ is called the \emph{tail} and vertex $v$ is called the \emph{head} of an arc $(u,v)$, and we say that $(u,v)$ is an arc \emph{outgoing} from $u$ and \emph{incoming} at $v$.

For a vertex $v$ in a digraph $G$, we let $N_{G}^+(v) = \{u \in V(G) \setminus  \{v\} : (v,u) \in A(G)\}$, $N_{G}^-(v) = \{w \in V(G) \setminus  \{v\} : (w,v) \in A(G)\}$, and $N_{G}(v) = N_{G}^+(v) \cup N_{G}^-(v)$, and we call these sets the \emph{out-neighborhood},  \emph{in-neighborhood}, and \emph{neighborhood} of $v$, respectively. The \emph{out-degree} (resp. \emph{in-degree}) of a vertex $v$ is the number of arcs outgoing from (resp. incoming at) $v$, and its \emph{degree} is the sum of its out-degree and its in-degree.
In all these notations, we may omit the subscripts if the digraph $G$ is clear from the context.


A \emph{subdigraph} of a digraph $G=(V,A)$ is a digraph $H=(V',A')$ such that $V' \subseteq V$ and $A' \subseteq A$. A \emph{path} from a vertex $u$ to a vertex $v$ in a digraph, also called \emph{$(u,v)$-path}, is a subdigraph obtained from an undirected path between $u$ and $v$ by orienting all edges toward $v$. The \emph{length} of a path is its number of arcs, and by an $\ell$-path we denote a path of length $\ell$. A $(v_1,v_k)$-path visiting vertices $v_1,v_2, \ldots, v_k$, in this order, is denoted by $(v_1, v_2, \ldots, v_k)$. A \emph{directed acyclic graph}, or DAG for short, is a digraph with no directed cycles. It is easy to prove that a digraph $G$ is a DAG if and only if there exists a total ordering of $V(G)$, called a \emph{topological ordering}, so that all arcs of $G$ go from smaller to greater vertices in this ordering.


For two positive integers $k$ and $\ell$, a \emph{$(k \times \ell)$-spindle} is the union of $k$ pairwise internally vertex-disjoint directed $(u,v)$-paths of length $\ell$  between two vertices $u$ and $v$, which are called the \emph{endpoints} of the spindle. More precisely, $u$ is called the \emph{tail} and $v$ the \emph{head} of a spindle. A
\emph{$2$-spindle} is any $(\ell_1,\ell_2)$-spindle with $\ell_1, \ell_2 \geq 1$.





For an undirected graph $G$, we denote by $\{u,v\}$ an edge between two vertices $u$ and $v$. A \emph{matching} in a graph is a set of pairwise disjoint edges. A vertex $v$ is \emph{saturated} by a matching $M$ if $v$ is an endpoint of one of the edges in $M$. In that case, we say that $v$ is \emph{$M$-saturated}. Given two matchings $M$ and $N$ in a graph, we let $M\triangle N$ denote their symmetric difference, that is, $M\triangle N = (M \setminus N)  \cup (N \setminus M)$.

\medskip
\noindent \textbf{Parameterized complexity}. We refer the reader to~\cite{DF13,FG06,Nie06,FPT-book} for basic background on parameterized complexity, and we recall here only some basic definitions.
A \emph{parameterized problem} is a decision problem whose instances are pairs $(x,k) \in \Sigma^* \times \mathbb{N}$, where $k$ is called the \emph{parameter}.
A parameterized problem is \emph{fixed-parameter tractable} ({\sf FPT}) if there exists an algorithm $\Acal$, a computable function $f$, and a constant $c$ such that given an instance $I=(x,k)$,
$\Acal$ (called an {\sf FPT} \emph{algorithm}) correctly decides whether $I \in L$ in time bounded by $f(k) \cdot |I|^c$. A parameterized problem is \emph{slice-wise polynomial} ({\sf XP}) if there exists an algorithm $\Acal$ and two computable functions $f,g$ such that given an instance $I=(x,k)$,
$\Acal$ (called an {\sf XP} \emph{algorithm}) correctly decides whether $I \in L$ in time bounded by $f(k) \cdot |I|^{g(k)}$.


%

Within parameterized problems, the class {\sf W}[1] may be seen as the parameterized equivalent to the class \np of classical optimization problems. Without entering into details (see~\cite{DF13,FG06,Nie06,FPT-book} for the formal definitions), a parameterized problem being {\sf W}[1]-\emph{hard} can be seen as a strong evidence that this problem is {\sl not} \fpt. The canonical example of {\sf W}[1]-hard problem is \textsc{Independent Set} parameterized by the size of the solution.
To transfer ${\sf W}[1]$-hardness from one problem to another, one uses a \emph{parameterized reduction}, which given an input $I=(x,k)$ of the source problem, computes in time $f(k) \cdot |I|^c$, for some computable function $f$ and a constant $c$, an equivalent instance $I'=(x',k')$ of the target problem, such that $k'$ is bounded by a function depending only on $k$. An equivalent definition of $\W$[1]-hard problem is any problem that admits a parameterized reduction from \textsc{Independent Set} parameterized by the size of the solution.


\medskip
\noindent \textbf{Matroids}. A pair $\mathcal{M} = (E, \mathcal{I})$, where $E$ is a ground set and $\mathcal{I}$ is a family of subsets of $E$,  is a \emph{matroid} if it satisfies the following three axioms:
\begin{enumerate}
\item $\emptyset \in \mathcal{I}$.
\item If $A' \subseteq A$ and $A \in \mathcal{I}$, then $A' \in \mathcal{I}$.
\item If $A,B \in \mathcal{I}$ and $|A| < |B|$, then there is $e \in B \setminus A$ such that $A \cup \{e\} \in \mathcal{I}$.
\end{enumerate}
The sets in $\mathcal{I}$ are called the \emph{independent sets} of the matroid. An inclusion-wise maximal set of $\mathcal{I}$ is called a \emph{basis} of the matroid. Using the third axiom, it is easy to show that all the bases of a matroid $\mathcal{M}$ have the same size, which is called the \emph{rank} of $\mathcal{M}$. A pair $\mathcal{M} = (E, \mathcal{I})$ over an $n$-element ground set $E$ is called a \emph{uniform matroid} if $\mathcal{I} = \{A \subseteq E : |A| \leq k\}$ for some for constant $k$.  For a broader overview on matroids, we refer to~\cite{Oxley}.

\medskip

For a positive integer $k$, we denote by  $[k]$ the set of all integers $i$ such that $1 \leq i \leq k$. Throughout the article, unless stated otherwise, we let $n$ denote the number of vertices of the input digraph of the problem under consideration.

\section{Complexity dichotomy in terms of the length of the paths}
\label{sec:dichomoty}

In this section we focus on the two natural optimization versions of finding subdivisions of spindles mentioned in the introduction, namely \textsc{Max $(k \times \bullet)$-Spindle Subdivision} and \textsc{Max $(\bullet \times \ell)$-Spindle Subdivision}.

It is easy to prove that the first problem is \NP-hard for any integer $k \geq 1$, by a simple reduction from \textsc{Longest Path}.

\begin{theorem}\label{thm:hard-max-ell}
Let $k \geq 1$ be a fixed integer. The \textsc{Max $(k \times \bullet)$-Spindle Subdivision} problem is \NP-hard.
\end{theorem}
\begin{proof} We provide a polynomial reduction from the \textsc{Longest Path} problem on general digraphs, which is \NP-hard as it generalizes \textsc{Hamiltonian Path}~\cite{GareyJ79comp}. For $k=1$, \textsc{Max $(k \times \bullet)$-Spindle Subdivision} is exactly the \textsc{Longest Path} problem, and the result follows. For $k >1$, let $G$ be an instance of \textsc{Longest Path} with $n$ vertices, and we build an instance $G'$ of \textsc{Max $(k \times \bullet)$-Spindle Subdivision} as follows. We start with $G$ and we add to it $2k-2$ new vertices $s_1,\ldots,s_{k-1}, t_1,\ldots,t_{k-1}$. For $i \in [k-1]$, we add an arc from every vertex of $G$ to $s_i$, and arc from $t_i$ to every vertex of $G$, and a path from $s_i$ to $t_i$ with $n$ edges through $n-1$ new vertices. This completes the construction of $G'$. It is clear that the length of a longest path in $G$ equals the largest integer $k$ such that $G'$ contains a subdivision of a $(k \times \ell)$-spindle, concluding the proof.
\end{proof}

We now present the complexity  dichotomy for the second problem, in order to prove Theorem~\ref{thm:dichotomy}. We start with the hardness result.


\begin{theorem}\label{thm:hard-max-k}
Let $\ell \geq 4$ be a fixed integer. The \textsc{Max $(\bullet \times \ell)$-Spindle Subdivision} problem is \NP-hard, even when restricted to DAGs.
\end{theorem}
\begin{proof} We provide a polynomial reduction from \textsc{3-Dimensional Matching}, which is \NP-hard~\cite{GareyJ79comp}.  In the \textsc{3-Dimensional Matching} problem, we are given three sets $A, B, C$ of the same size and a set of triples $\mathcal{T} \subseteq A \times B \times C$. The objective is to decide whether there exists a set $\mathcal{T'} \subseteq \mathcal{T}$ of pairwise disjoint triples with $|\mathcal{T'}| = |A|$. Given an instance $(A,B,C,\mathcal{T})$ of \textsc{3-Dimensional Matching}, with $|A|=n$ and $\mathcal{T}=m$, we construct an instance $G$ of \textsc{Max $(\bullet \times \ell)$-Spindle Subdivision} as follows.
We first present the reduction for $\ell = 4$, and then we explain how to modify it for a general $\ell > 4$.

For every $i \in [n]$, we add to $G$ three vertices $a_i,b_i,c_i$, corresponding to the elements in the sets $A,B,C$, respectively. Let $H$ the digraph with vertices $x_0,x_1,y_0,y_1,z_0,z_1,a,b,c$ and arcs $(x_0,x_1), (x_1,a), (x_1,y_0), (y_0,y_1), (y_1,b), (x_0,z_0), (z_0,z_1), (z_1,c)$  (see Figure~\ref{fig:reduction-4}(a)). For every triple $T \in \mathcal{T}$, with $T = (a_i,b_j,c_p)$, we add to $G$ a copy of $H$ and we identify vertex $a$ with $a_i$, vertex $b$ with $b_j$, and vertex $c$ with $c_p$. Finally, we add a new vertex $s$ that we connect to all other vertices introduced so far, and another vertex $t$ to which we connect all other vertices introduced so far except $s$.

\begin{figure}[h!]
\begin{center}
\includegraphics[width=0.93\textwidth]{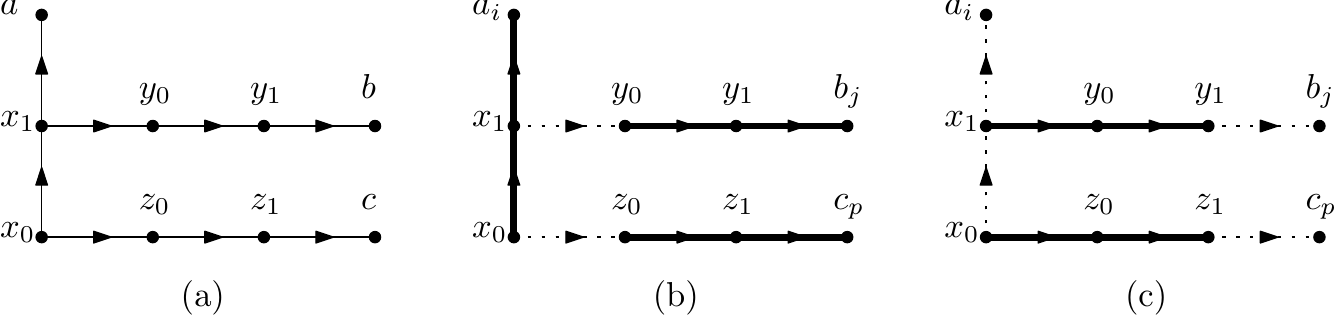}
\end{center}\vspace{-.15cm}
\caption{(a) Digraph $H$. (b) Selected paths when $T \in \mathcal{T}'$. (c) Selected paths when $T \in \mathcal{T} \setminus \mathcal{T'}$.}
\label{fig:reduction-4}
\end{figure}

 The constructed digraph $G$  is easily seen to be a DAG. Indeed, we can define a topological ordering of $V(G)$ so that all arcs go from left to right as follows. We select $s$ (resp. $t$) as the leftmost (resp. rightmost) vertex. We divide the remaining vertices of $G$ into two blocks. On the right, we place all the vertices $\{a_i,b_i,c_i : i \in [n]\}$, and we order them arbitrarily. On the left, we place the remaining vertices of $G$, which we also order arbitrarily, except that for every triple $T \in \mathcal{T}$, we order the vertices in its copy of $H$, distinct from $a,b,c$, such that $x_0 < x_1 <y_0 < y_1 < z_0 < z_1$ holds. One can check that, with respect to this ordering, all the arcs of $G$ go from left to right.

Note that $|V(G)| = 3n + 6m + 2$, and therefore the largest integer $k$ for which $G$ contains a subdivision of a $(k \times 4)$-spindle is $k^* := n + 2m$, as each path involved in such a spindle contains at least three vertices distinct from its endpoints. We claim that $(A,B,C,\mathcal{T})$ is a \textsc{Yes}-instance of \textsc{3-Dimensional Matching} if and only if $G$ contains a subdivision of a $(k^* \times 4)$-spindle.

Suppose first that $(A,B,C,\mathcal{T})$ is a \textsc{Yes}-instance, and let $\mathcal{T'} \subseteq \mathcal{T}$ be a solution. We proceed to define a set $\mathcal{P}$ of $n+2m$ vertex-disjoint 2-paths in $G \setminus \{s,t\}$, which together with $s$ and $t$ yield the desired spindle. For every $T \in \mathcal{T'}$, with $T = (a_i,b_j,c_p)$, we add to $\mathcal{P}$ the three paths $(x_0,x_1,a_i)$, $(y_0,y_1,b_j)$, and $(z_0,z_1,c_p)$ (see the thick arcs in Figure~\ref{fig:reduction-4}(b)). On the other hand, for every $T \in \mathcal{T} \setminus \mathcal{T'}$, with $T = (a_i,b_j,c_p)$, we add to $\mathcal{P}$ the two paths $(x_1,y_0,y_1)$ and $(x_0,z_0,z_1)$ (see the thick arcs in Figure~\ref{fig:reduction-4}(c)). Since $\mathcal{T}'$ is a solution of \textsc{3-Dimensional Matching}, it holds that $|\mathcal{T}'| = n$, and thus $\mathcal{P} = 3n + 2(m - n) = n + 2m$, as required.

Conversely, suppose that $G$ contains a subdivision of a $(k^* \times 4)$-spindle $S$. Since $s$ and $t$ are the only vertices in $G$ with in-degree and out-degree at least $k^*$, respectively, necessarily they are the endpoints of $S$. Since $|V(G) \setminus \{s,t\} |= 3k^*$, it follows that $S \setminus \{s,t\}$ consists of a collection $\mathcal{P}$ of $k^*$ vertex-disjoint 2-paths that covers all the vertices in $V(G) \setminus \{s,t\}$. Let $H$ be the subdigraph in $G$ associated with an arbitrary triple $T \in \mathcal{T}$, and consider $\mathcal{P} \cap H$. By construction of $H$, it follows that if $\mathcal{P} \cap H$ is not equal to one of the configurations corresponding to the thick arcs of Figure~\ref{fig:reduction-4}(b) or Figure~\ref{fig:reduction-4}(c), necessarily at least one vertex in $V(H)$ would not be covered by $\mathcal{P}$, a contradiction. Let $\mathcal{T}'$ be the set of triples in $\mathcal{T}$ such that the corresponding gadget $H$ intersects $\mathcal{P}$ as in Figure~\ref{fig:reduction-4}(b). It follows that $3|\mathcal{T}'| + 2(m-|\mathcal{T}'|) = |\mathcal{P}| = k^* = n + 2m$, and therefore $|\mathcal{T}'| = n$. Since all the 2-paths in $\mathcal{P}$ associated with the triples in $\mathcal{T}'$ are vertex-disjoint, we have that $\mathcal{T}'$ is a collection of $n$  pairwise disjoint triples, hence a solution of \textsc{3-Dimensional Matching}.

For a general $\ell > 4$, we define the digraph $G$ in the same way, except that we subdivide the arcs outgoing from $s$ exactly $\ell - 4$ times. The rest of the proof is essentially the same, and the result follows.\end{proof}

We now turn to the cases that can be solved in polynomial time. We first need some ingredients to deal with the case $\ell = 3$, which is the most interesting one. Let $G$ be a digraph and let $X$ and $Y$ be two subsets of $V(G)$.
We say that a path $P$ is \emph{directed from $X$ to $Y$} if $P$ is a directed path with first vertex $x$ and last vertex $y$ such that $x\in X$ and $y\in Y$. The path $P$ is \emph{nontrivial} if its endpoints are distinct.

The following proposition will be the key ingredient in the proof of Theorem~\ref{thm:poly-max-k}. Its proof is inspired by similar constructions given by Schrijver~\cite{Schrijver01} on undirected graphs and by Kriesell~\cite{Kri05} on directed graphs, usually called \emph{vertex splitting procedure} (see~\cite[Section 5.9]{BJ-Gutin-book}).

\begin{proposition}\label{prop:Tpaths}
Let $G$ be a digraph and let $X$ and $Y$ be two subsets of $V(G)$. The maximum number of vertex-disjoint directed nontrivial paths from $X$ to $Y$ can be computed in polynomial time.
\end{proposition}
\begin{proof} Let $\mathcal{P}$ be any collection of vertex-disjoint directed nontrivial paths from $X$ to $Y$ in $G$.
We can rebuild each path in $\mathcal{P}$ so that it has no internal vertices in $X\cup Y$.
Therefore, we can assume $G$ has no arcs to a vertex in $X\setminus Y$ or from a vertex in $Y\setminus X$.

Let $G'$ be the undirected graph built from $G$ as follows.
The vertex set of $G'$ is obtained from $V(G)$ by adding a copy $v'$ of each vertex $v$ not in $X\cup Y$.
We build the edge set of $G'$ starting from the empty set as follows.
For every vertex $v$ not in $X\cup Y$, add the edge $\{v,v'\}$.
For each arc $(u,v)$ in $G$,
we add the edge $\{u,v\}$ if $v\in X\cup Y$ and the edge $\{u,v'\}$ otherwise. See Figure~\ref{fig:example-matching}(a)-(b) for an example.

\begin{figure}[h!]
\begin{center}
\includegraphics[width=.94\textwidth]{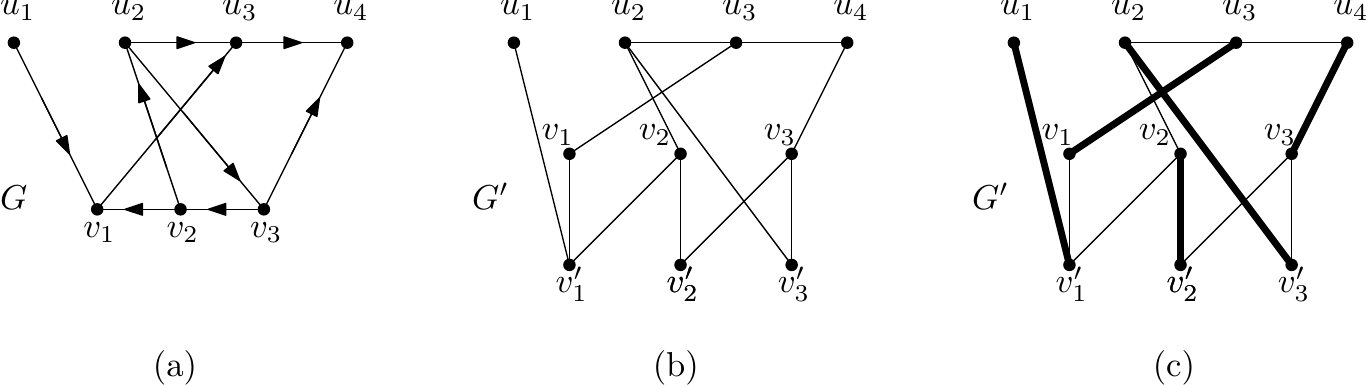}
\end{center}\vspace{-.15cm}
\caption{(a) Digraph $G$ with $X= \{u_1,u_2,u_3\}$ and $Y= \{u_3,u_4\}$. (b) Graph $G'$ associated with $G$. (c) The thick edges define a matching of size five in $G'$, corresponding to the two vertex-disjoint directed nontrivial paths $(u_1,v_1,u_3)$ and $(u_2,v_3,u_4)$ from $X$ to $Y$ in $G$.}
\label{fig:example-matching}
\end{figure}

\begin{claim}\label{claim-Tpaths}
The digraph $G$ contains a family of $k$ vertex-disjoint directed nontrivial paths from $X$ to $Y$ if and only if $G'$ has a matching of size $k+|V(G)\setminus (X\cup Y)|$.
\end{claim}
\begin{proof-claim}
Let $\mathcal{P}$ be a family of $k$ vertex-disjoint directed nontrivial paths from $X$ to $Y$ in $G$.
We may assume every path in $\mathcal{P}$ has no internal vertices in $X\cup Y$.
Let $U$ be the subset of vertices of $V(G)\setminus (X\cup Y)$ that are not in a path in $\mathcal{P}$.

We build a matching $M$ of $G'$ starting with $M = \{\{u,u'\} : u\in U\}$ as follows.
For every arc $(u,v)$ used in some path of $\mathcal{P}$, we add $\{u,v\}$ to $M$ if $v\in X\cup Y$, and $\{u,v'\}$ otherwise (see Figure~\ref{fig:example-matching}(c)).
Note that $M$ is indeed a matching, as vertices in $X\cup Y$ appear in at most one arc on a path in $\mathcal{P}$.
For a vertex $v$ not in $X\cup Y$, $v$ appears at most once as an internal vertex in a path $P$ of $\mathcal{P}$.
Therefore, it appears in exactly two arcs of $P$ and exactly once in an arc to $v$ and once in an arc from $v$.

We now claim that the number of $M$-saturated vertices in $G'$ is $2(k + |V(G)\setminus (X\cup Y)|)$.
This claim implies that $M$ has $k + |V(G)\setminus (X\cup Y)|$ edges.
To prove this claim, first note that all vertices in $V(G)\setminus (X\cup Y)$ are saturated.
Indeed, if $v$ is in $U$, then both $v$ and $v'$ are initially saturated.
Otherwise, $v$ is an internal vertex of a path in $\mathcal{P}$ and is contained in two edges that saturate both $v$ and $v'$.
To conclude, note that every path in $\mathcal{P}$ contains exactly two vertices in $X\cup Y$, namely its endpoints, and, therefore, saturates exactly two vertices of $X\cup Y$ in $G'$.

Now, let $M$ be a matching of $G'$ of size $k+|V(G)\setminus (X\cup Y)|$.
Let $N$ be the matching $\{\{v,v'\} : v\in V(G)\setminus (X\cup Y)\}$ and $H = G[M\triangle N]$.
Since $|M| = k + |N|$, $H$ contains at least $k$ components with more edges in $M$ than in $N$.
We claim that from these components we can obtain $k$ vertex-disjoint nontrivial paths in $G$.

To prove this claim, let $C$ be a component of $H$ with more edges in $M$ than in $N$.
Since $C$ has more edges in $M$, then it is a path alternating between edges of $M$ and $N$ that starts and ends with an edge of $M$, its endpoints are $N$-unsaturated and its internal vertices are $N$-saturated.
Thus, the endpoints of $C$ are its only vertices in $X\cup Y$.
Also note that if a vertex $w$ of $V(C)\cap V(G)$ is not in $X\cup Y$, then both $w$ and $w'$ are in $C$ and neither $w$ nor $w'$ appear in any other component of $H$.

Let $u$ and $v$ be the endpoints of $C$ and the set $W$ of internal vertices of $C$ that are also in $V(G)$ be $\{w_1, \ldots, w_\ell\}$.
Note that $u\neq v$ as $C$ contains at least one edge in $M$.
If $W = \emptyset$, then $uv$ is an edge of $G'$ and assume the edge $\{u,v\}$  is directed in $G$ from $u$ to $v$.
If $W\not = \emptyset$, then assume the transversal of $C$ from $u$ to $v$ visits the vertices in the order $u, w_1',w_1, w_2', w_2, \ldots, w_\ell', w_\ell, v$.
In both cases, note that $u, w_1, \ldots, w_\ell, v$ is the transversal of a directed path from $u$ to $v$ in $G$.
Since $G$ has no edge leaving a vertex of $Y\setminus X$ and no edge going into a vertex of $X\setminus Y$, then $u\in X$ and $v\in Y$.
\end{proof-claim}

\igjournal{saturated vertices are in a component, under the condition that we take a maximum matching $M$ minimizing the number of saturated vertices in $X \cup Y$}

Claim~\ref{claim-Tpaths} tells us that we can obtain a maximum number of vertex-disjoint nontrivial paths from $X$ to $Y$ in $G$ by finding a maximum matching in the graph $G'$, which can be done in polynomial time~\cite{Diestel12}. The proposition follows.
\end{proof}

We are now ready to prove the main algorithmic result of this section.

\begin{theorem}\label{thm:poly-max-k}
Let $\ell \leq 3$ be a fixed integer. The \textsc{Max $(\bullet \times \ell)$-Spindle Subdivision} problem can be solved in polynomial time.
\end{theorem}
\begin{proof}
If $\ell = 1$, then the problem can be solved just by computing a maximum flow between every pair of vertices of the input digraph, which can be done in polynomial time~\cite{BJ-Gutin-book}. If $\ell = 2$, we use the same algorithm, except that for every pair of vertices we first delete all the arcs between them before computing a maximum flow, as the paths of length one are the only forbidden ones in a subdivision of a $(k \times 2)$-spindle.

Let us now focus on the case $\ell = 3$. We first guess a pair of vertices $s$ and $t$ of $V(G)$ as candidates for being the tail and head of the desired spindle, respectively, and we delete the arcs between $s$ and $t$, if any. The crucial observation is that the largest integer $k$ such that $G$ contains a $(k \times 3)$-spindle having $s$ and $t$ as tail and head, respectively,  equals the maximum number of vertex-disjoint directed nontrivial paths from $N^+(s)$ to $N^-(t)$ in the digraph $G \setminus \{s,t\}$. Now the result follows directly by applying the polynomial-time algorithm given by Proposition~\ref{prop:Tpaths} with input graph $G \setminus \{s,t\}$, $X=N^+(s)$, and $Y=N^-(t)$.
\end{proof}


We now prove a generalization of the result given in Theorem~\ref{thm:poly-max-k}, but using a much more powerful tool. Namely, instead of reducing the problem to finding a matching of appropriate size in an auxiliary graph, as in the proof of Proposition~\ref{prop:Tpaths}, we use as a black box an algorithm of Lov{\'{a}}sz~\cite{Lovasz80} to solve the \emph{matroid matching} problem in polynomial time for linearly-represented matroids (see~\cite{Oxley} for any missing definition about matroids).


\begin{theorem}\label{thm:poly-matroid-matching}
Given a digraph $G$ and three non-negative integers $k_1,k_2,k_2$, deciding whether $G$ contains a subdivision of a $(\ell^1_1, \ldots, \ell^1_{k_1}, \ell^2_1, \ldots, \ell^2_{k_2},\ell^3_1, \ldots, \ell^3_{k_3})$-spindle such that, for $j \in [3]$ and $i \in [k_j]$, $\ell^j_i = j$, can be solved in polynomial time.
\end{theorem}


\begin{proof}
We iterate on pairs of vertices $s$ and $t$ in $G$ to decide if the desired spindle exists with tail $s$ and head $t$.
From now on, we consider a fixed pair of vertices $s$ and $t$.
If a $(\ell^1_1, \ldots, \ell^1_{k_1}, \ell^2_1, \ldots, \ell^2_{k_2},\ell^3_1, \ldots, \ell^3_{k_3})$-spindle subdivision exists with tail $s$ and head $t$, let $S$ be one such subdivision.

Let $p$ be the number of arcs with tail $s$ and head $t$.
Note that if $S$ exists, it can use at most $\min\{p, k_1\}$ arcs between $s$ and $t$.
In fact, we can assume $S$ uses exactly $\min\{p, k_1\}$ arcs between $s$ and $t$ as, otherwise, there is a 1-path which was subdivided and can be changed to an unused arc from $s$ to $t$.
All other 1-paths of the spindle must have been subdivided and have length at least two in $S$.
Therefore, $S$ exists if and only if there is a $(\ell^2_1, \ldots, \ell^2_{k'_2},\ell^3_1, \ldots, \ell^3_{k_3})$-spindle subdivision with tail $s$ and head $t$ with $k'_2 = k_2 + k_1 - \min\{p, k_1\}$.
From now on, assume $k_1 = 0$.

Let $X = N^+(s)$ and $Y = N^-(t)$.
For a non-negative integer $r$, let $m_r$ be the maximum number of nontrivial vertex-disjoint paths from $X$ to $Y$ such that at least $r$ vertices of $X\cap Y$ are not used by these paths.
We claim that $S$ exists if and only if $r+m_r \ge k_2+k_3$ for some integer $r$ with $0\le r \le \min\{k_2, |X\cap Y|\}$.
If $r+m_r \ge k_2+k_3$, then we can find $S$ by joining $s$ and $t$ to nontrivial paths from $X$ to $Y$ to find $k_2+k_3 - r$ paths of length at least three.
The remaining paths are built as 2-paths through $r$ vertices in $X\cap Y$ not used by the nontrivial paths.
On the other hand, if $S$ exists, let $r$ be the number of its 2-paths which are not subdivided in $G$, and note that $0\le r \le \min\{k_2, |X\cap Y|\}$.
These $r$ paths each contain $s$, $t$, and a vertex of $X\cap Y$.
The $k_2 - r$ 2-paths  which are subdivided in $G$ have length at least 3 in $G$.
Thus, by deleting $s$ and $t$ from $S$, we have $k_2 - r + k_3$ nontrivial vertex-disjoint paths from $X$ to $Y$ disjoint from the $r$ vertices used 2-paths of $S$ in $G$.
Therefore, we have $m_r \ge k_2 - r + k_3$, which is equivalent to $r+ m_r \ge k_2 + k_3$.

We iterate on values of $r$ with $0\le r\le \min\{k_2, |X\cap Y|\}$ to decide if $S$ exists.
We finish this proof by showing how to find the value of $m_r$ in polynomial time for a fixed integer $r$.
Namely, we use the matroid matching algorithm of Lov{\'{a}}sz~\cite{Lovasz80}.
For a graph $H$ and a linear matroid $M$ over $V(H)$, this algorithm finds  in polynomial time a maximum matching in $G$ whose saturated vertices form an independent set of $M$.

For the graph $H$ above, we use the graph $G'$ built in Proposition~\ref{prop:Tpaths}.
Recall that Claim~\ref{claim-Tpaths}  shows that a matching in $H$ of size $k + |V(G)\setminus (X\cup Y)|$ exists if and only if there are $k$ vertex-disjoint directed nontrivial paths from $X$ to $Y$. Furthermore, it can be also proved that the saturated vertices in $X\cup Y$ by the matching correspond precisely to the endpoints of the nontrivial paths from $X$ to $Y$\igjournal{I noticed this should be added to Claim~\ref{claim-Tpaths} (just a little more work). Actually, I think the description of $G'$ could be removed from the proof of Proposition~\ref{prop:Tpaths} so that it makes more sense for it to be referenced here. The claim can be used by the proposition and this theorem.}.
The linear matroid $M$ is such that a set of vertices of $H$ is independent in $M$ if it is disjoint from at least $r$ vertices of $X\cap Y$.
To see that $M$ is a linear matroid, note that $M$ is the dual of the $r$-uniform matroid over $X\cap Y$ by extending its ground set to $V(H)$ without changing the independent sets.

Now, a matching in $H$ of size $k + |V(G)\setminus (X\cup Y)|$ that is independent in $M$ corresponds  precisely to $k$ vertex-disjoint directed nontrivial paths from $X$ to $Y$ that are disjoint from a set of $r$ vertices of $X\cap Y$, and the matroid matching algorithm~\cite{Lovasz80} can find the value of $m_r$.
\end{proof}

\section{Finding subdivisions of 2-spindles}
\label{sec:two-paths}

In this section we focus on finding subdivisions of 2-spindles, and we prove Theorem~\ref{thm:2blockFPTsum} and Theorem~\ref{thm:2blockFPT-2length}. We prove the negative and the positive results of both theorems separately. Namely, we provide the hardness results in Section~\ref{sec:hardness-two-paths} and we focus on the
{\sf FPT} algorithms in Section~\ref{sec:FPT-algo}.

\subsection{Hardness results}
\label{sec:hardness-two-paths}

We start by proving the {\sf NP}-hardness results.

\begin{proposition}\label{prop:two-blocks-hard}
The \textsc{Max $(\bullet,\bullet)$-Spindle Subdivision} problem is {\sf NP}-hard. For every fixed integer $\ell_1 \geq 1$, the \textsc{Max $(\ell_1,\bullet)$-Spindle Subdivision} problem is {\sf NP}-hard.
\end{proposition}
\begin{proof}
For both problems, we present a reduction from the \textsc{Directed Hamiltonian $(s,t)$-Path} problem, which consists in, given a digraph $G$ and two vertices $s,t \subseteq V(G)$, deciding whether $G$ has an $(s,t)$-path that is Hamiltonian. This problem is easily seen to be {\sf NP}-hard by a simple reduction from \textsc{Directed Hamiltonian Cycle}, which is known to be {\sf NP}-hard~\cite{GareyJ79comp}: given an instance $G$ of \textsc{Directed Hamiltonian Cycle}, construct from $G$ an instance $G'$ of \textsc{Directed Hamiltonian $(s,t)$-Path} by choosing an arbitrary vertex $v \in V(G)$ and splitting it into two vertices $s$ and $t$ such that $s$ (resp. $t$) is incident to exactly those arcs in $G$ that were outgoing from (resp. incoming at) $v$.

We first prove the hardness of \textsc{Max $(\bullet,\bullet)$-Spindle Subdivision}. Given an instance $G$ of \textsc{Directed Hamiltonian $(s,t)$-Path}, with $|V(G)|=n$, build an instance $G'$ of \textsc{Max $(\bullet,\bullet)$-Spindle Subdivision} as follows. Start from $G$, and delete all the arcs incoming at $s$ or outgoing from $t$, if any. Finally, add a new vertex $v$ and arcs $(s,v)$ and $(v,t)$.
We claim that $G$ has a Hamiltonian $(s,t)$-path if and only if $G'$ contains a subdivision of a $(\ell_1,\ell_2)$-spindle with $\min\{\ell_1,\ell_2\} \geq 1$ and $\ell_1 + \ell_2 = n+1$. Assume first that $G$ has a Hamiltonian $(s,t)$-path $P$. Then $G'$ contains a $(2,n-1)$-spindle defined by the $2$-path $(s,v,t)$ together with the Hamiltonian $(s,t)$-path $P$. Conversely, assume that $G'$ contains a subdivision $S$ of a $(\ell_1,\ell_2)$-spindle with $\min\{\ell_1,\ell_2\} \geq 1$ and $\ell_1 + \ell_2 = n+1$. Suppose that the newly added vertex $v \in V(G')$ does not belong to $S$, which implies that $|V(S)| \leq |V(G)| = n$. Since a $(\ell_1,\ell_2)$-spindle contains exactly $\ell_1 + \ell_2$ vertices, if follows that $ |V(S)| \geq \ell_1 + \ell_2 = n+1$, a contradiction to the previous sentence. Therefore, $v \in V(S)$ and so $(s,v,t)$ is one of the two paths of $S$. Thus, the remaining path of $S$ is an $(s,t)$-path of length $n-1$ in $G$, that is, a Hamiltonian $(s,t)$-path  in $G$.

We now prove the hardness of \textsc{Max $(\ell_1,\bullet)$-Spindle Subdivision} for every fixed integer $\ell_1 \geq 1$. Given an instance $G$ of \textsc{Directed Hamiltonian $(s,t)$-Path}, with $|V(G)|=n$, build an instance $G'$ of \textsc{Max $(\ell_1,\bullet)$-Spindle Subdivision} as follows. Start from $G$, and delete all the arcs incoming at $s$ (resp. outgoing from $t$), if any, and the arc $(s,t)$, if it exists. Finally, add an $(s,t)$-path with $\ell_1$ arcs consisting of new vertices and arcs.  One can easily check that $G$ has a Hamiltonian $(s,t)$-path if and only if $G'$ contains a subdivision of a $(\ell_1,n-1)$-spindle.
\end{proof}


Bj{\"{o}}rklund et al.~\cite{BjorklundHK04} showed that assuming
the Exponential Time Hypothesis\footnote{The {\sf ETH} states that there is no algorithm solving \textsc{3-SAT} on a formula with $n$ variables in time $2^{o(n)}$.} ({\sf ETH}) of Impagliazzo et al.~\cite{ImpagliazzoPZ01}, the  \textsc{Directed Hamiltonian Cycle} problem cannot be solved in time $2^{o(n)}$. This result together with the proof of Proposition~\ref{prop:two-blocks-hard} directly imply the following two results assuming the {\sf ETH}, claimed in
Theorem~\ref{thm:2blockFPTsum} and Theorem~\ref{thm:2blockFPT-2length}, respectively. The first one is that, given a digraph $G$ and a positive integer $\ell$, the problem of deciding whether there exist two strictly positive integers $\ell_1,\ell_2$ with $\ell_1+\ell_2 = \ell$ such that $G$ contains a subdivision of a $(\ell_1,\ell_2)$-spindle cannot be solved in time $2^{o(\ell)}\cdot n^{O(1)}$. The second one is that, given a digraph $G$ and two integers $\ell_1, \ell_2$ with $\ell_2 \geq \ell_1 \geq 1$, the problem of deciding whether $G$ contains a subdivision of a $(\ell_1,\ell_2)$-spindle cannot be solved in time $2^{o(\ell_2)}\cdot n^{O(\ell_1)}$.

Concerning the existence of polynomial kernels, it is easy to prove that none of the above problems  admits polynomial kernels unless ${\sf NP} \subseteq {\sf coNP} / {\sf poly}$. Indeed, taking the disjoint union of $t$ instances of any of these two problems defines a \emph{cross-composition}, as defined by Bodlaender et al.~\cite{BJK14}, from the problem to itself, directly implying the desired results as both problems are {\sf NP}-hard by Proposition~\ref{prop:two-blocks-hard}.
We refer to~\cite{BJK14} for the missing definitions.

\subsection{FPT algorithms}
\label{sec:FPT-algo}

Our {\sf FPT} algorithms for finding subdivisions of $(\ell_1,\ell_2)$-spindles are based on the technique of \emph{representative families} introduced by Monien~\cite{Monien85}. We use the improved version of this technique recently presented by
Fomin et al.~\cite{FominLPS16} and, more precisely, our algorithms and notation are inspired by the ones for \textsc{Long Directed Cycle} given in~\cite{FominLPS16}.
We start with some definitions introduced from~\cite{FominLPS16} that can also be found in~\cite{FPT-book}.


Two independent sets $A,B$ of a matroid $\mathcal{M}$ \emph{fit} if $A \cap B = \emptyset$ and $A \cup B$ is independent.

\begin{definition}\label{def:represent}
 Let $\mathcal{M}$ be a matroid and $\mathcal{A}$ be a family of sets of size $p$ in $\mathcal{M}$. A subfamily $\mathcal{A}' \subseteq \mathcal{A}$ is said to \emph{$q$-represent} $\mathcal{A}$ if for every set $B$ of size $q$ such that there is an $A \in \mathcal{A}$ that fits $B$, there is an $A' \in \mathcal{A}'$ that also fits $B$. If $\mathcal{A}'$ $q$-represents $\mathcal{A}$, we write $\mathcal{A}' \subseteq^q_{\text{rep}} \mathcal{A}$.
\end{definition}

\subsubsection{Finding 2-spindles with large total size}
\label{sec:FPT-1}

We start with the algorithm to solve the problem of, given a digraph $G$ and a positive integer $\ell$, deciding whether there exist two strictly positive integers $\ell_1,\ell_2$ with $\ell_1+\ell_2 = \ell$ such that $G$ contains a subdivision of a $(\ell_1,\ell_2)$-spindle, running in time $2^{O(\ell)}\cdot n^{O(1)}$.


If a subdigraph $S$ of $G$ is a subdivision of a $(\ell_1,\ell_2)$-spindle, with $\min\{\ell_1,\ell_2\} \geq 1$ and $\ell_1+\ell_2 = \ell$, we say that $S$ is a \emph{good spindle}. We may assume that $\max\{\ell_1,\ell_2\} \geq 2$, as otherwise the desired spindle is just an arc with multiplicity two, which can be detected in polynomial time by using a maximum flow algorithm.

The following simple observation, whose proof can be easily verified, will be crucially used by the algorithm that we propose in the sequel. See Figure~\ref{fig:lem:structural-1} for an illustration.

\begin{lemma}\label{lem:structural-1}
A digraph $G$ has a good spindle if and only if there exist vertices $u,u_1,u_2,v$, integers $\ell_1,\ell_2$ with $\min\{\ell_1,\ell_2\} \geq 1$ and $\ell_1+\ell_2 = \ell$,
a $(u,u_1)$-path $P^u_1$ on $\ell_1$ vertices, a $(u,u_2)$-path $P^u_2$ on $\ell_2$ vertices, a $(u_1,v)$-path $P^v_1$, and a $(u_2,v)$-path $P^v_2$ such that $V(P^u_1) \cap V(P^u_2) = \{u\}$, $V(P^v_1) \cap V(P^v_2) = \{v\}$, $V(P^u_1) \cap V(P^v_1) = \{u_1\}$, $V(P^u_2) \cap V(P^v_2) = \{u_2\}$, and, if $\min\{\ell_1,\ell_2\} \geq 2$, $V(P^u_1) \cap V(P^v_2) = V(P^u_2) \cap V(P^v_1) = \emptyset$.
 \end{lemma}

In the above lemma, note that if $\min\{\ell_1,\ell_2\} = 1$ then one of the paths $P^u_1$ and $P^u_2$, say $P^u_1$, may be degenerate to vertex $u$, and in that case we have that $u_1 = u$.

\begin{figure}[h!]
\begin{center}
\includegraphics[width=.45\textwidth]{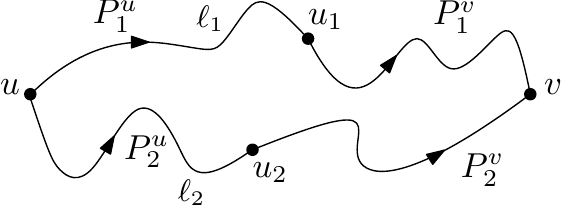}
\end{center}\vspace{-.15cm}
\caption{Illustration of the vertices and paths described in Lemma~\ref{lem:structural-1}.}
\label{fig:lem:structural-1}
\end{figure}

Motivated by Lemma~\ref{lem:structural-2}, for every triple of vertices $u,u_1,u_2 \in V(G)$ and positive integers $\ell_1, \ell_2$, we define
\begin{align*}
\mathcal{S}_{u,u_1,u_2}^{\ell_1,\ell_2} =  \Big\{ X :  \ \  & S\subseteq V(G), |X| = \ell_1 + \ell_2 -1, \text{ and $G[X]$ contains a}\\
& \text{$(u,u_1)$-path $P_1^u$ on $\ell_1$ vertices and a $(u,u_2)$-path $P_2^u$}\\
& \text{on $\ell_2$ vertices such that $V(P_1^u) \cap V(P_2^u) = \{u\}$}\Big\}.
\end{align*}


The key idea is to compute efficiently a small family of subsets of $V(G)$ that {\sl represents} the above sets, which are too large for our purposes. More precisely, for every triple of  vertices $u,u_1,u_2 \in V(G)$ and positive integers $\ell_1, \ell_2, q$ with $\ell_1,\ell_2 \leq \ell$ and $q \leq 2 \ell - (\ell_1 + \ell_2)$,
we will compute in time $2^{O(\ell)}\cdot n^{O(1)}$ a $q$-representative family
$$
\widehat{\mathcal{S}}_{u,u_1,u_2}^{\ell_1,\ell_2,q} \subseteq^q_{\text{rep}} \mathcal{S}_{u,u_1,u_2}^{\ell_1,\ell_2}.
$$

As in~\cite{FominLPS16}, the matroid with respect to which we will define the above $q$-representative family $\widehat{\mathcal{S}}_{u,u_1,u_2}^{\ell_1,\ell_2,q}$ is the uniform matroid with ground set $V(G)$ and rank $\ell+q$.


We postpone the computation of the above $q$-representative families in time $2^{O(\ell)}\cdot n^{O(1)}$ to Section~\ref{sec:compute-representatives}, and assume now that we already have these families at hand. The following lemma states that they are enough to find the desired good spindle.

\begin{lemma}\label{lem:key-lemma}
If $G$ contains a good spindle, then there exist vertices $u,u_1,u_2,v$, integers $\ell_1,\ell_2$ with $\min\{\ell_1,\ell_2\} \geq 1$ and $\ell_1 + \ell_2 = \ell $, a set $\widehat{S}_u \in \widehat{\mathcal{S}}_{u,u_1,u_2}^{\ell_1,\ell_2,q}$ with $q \leq \ell -1$, a $(u_1,v)$-path $P_1^v$, and a $(u_2,v)$-path $P_2^v$ such that $V(P_1^v) \cap V(P_2^v) = \{v\}$ and $\widehat{S}_u \cap (V(P_1^v) \cup V(P_2^v)) = \{u_1,u_2\}$.
\end{lemma}
\begin{proof} Let $S$ be a good spindle in $G$ with minimum number of vertices, which exists by hypothesis, and let $u$ and $v$ be the tail and the head of $S$, respectively. Let $P_1^u = (u, \ldots, u_1)$ and $P_2^u = (u, \ldots, u_2)$ be two subdipaths in $S$ outgoing from $u$, on $\ell_1$ and $\ell_2$ vertices, respectively, with $\ell_1 + \ell_2 = \ell$. Let also $P_1^v = (u_1, \ldots, v)$ and $P_2^v = (u_2, \ldots, v)$ be the two subdipaths in $S$ from $u_1$ and $u_2$ to $v$, respectively (see Figure~\ref{fig:lem:structural-1}). Let $S_u = V(P_1^u) \cup V(P_2^u)$, and note that $S_u \in \mathcal{S}_{u,u_1,u_2}^{\ell_1,\ell_2}$.

In order to apply the properties of $q$-representative families, we define a vertex set $B \subseteq V(S)$ as follows. If $|V(S) \setminus S_u|  \leq \ell-2$, let $B = V(S) \setminus S_u$. Otherwise, let $B$ be the union of two subdipaths $P^B_1 = (v_1, \ldots, v)$ and $P^B_2 = (v_2, \ldots, v)$ in $S$ with $V(P^B_1) \cap V(P^B_1) = \{v\}$ and $|V(P^B_1) \cup V(P^B_1)| = \ell - 1$. Note that there may be several choices for the lengths of $P^B_1$ and $P^B_2$, as far as their joint number of vertices is equal to $\ell - 1$. Note also that $P^B_1$ (resp. $P^B_2$) is a subdipath of $P^v_1$ (resp. $P^v_2$).

Let $q = |B| \leq \ell -1$. Since $S_u \in \mathcal{S}_{u,u_1,u_2}^{\ell_1,\ell_2}$ and $S_u \cap B = \emptyset$, by definition of $q$-representative family there exists $\widehat{S}_u \in \widehat{\mathcal{S}}_{u,u_1,u_2}^{\ell_1,\ell_2,q}$ such that $\widehat{S}_u \cap B = \emptyset$. We claim that $\widehat{S}_u \cap (V(P_1^v) \cup V(P_2^v)) = \{u_1,u_2\}$, which concludes the proof of the lemma. If $|B|  \leq \ell-2$, the claim follows easily as $\widehat{S}_u \cap B = \emptyset$ and $B$ contains all the vertices in $V(S) \setminus S_u$. Suppose henceforth that $|B|  \geq \ell-1$, and let $\widehat{P}_1^u$ and $\widehat{P}_2^u$ be the two paths in $G[\widehat{S}_u]$ with $V(\widehat{P}_1^u) \cap V(\widehat{P}_1^u)= \{u\}$. Assume for contradiction that $(\widehat{S}_u \cap (V(P_1^v) \cup V(P_2^v))) \setminus  \{u_1,u_2\} \neq \emptyset$, and we distinguish two cases.

Suppose first that each of the paths $\widehat{P}_1^u$ and $\widehat{P}_2^u$ intersects exactly one of the paths $P_1^v$ and $P_2^v$. By hypothesis, there exists a vertex $w \in (\widehat{S}_u \cap (V(P_1^v) \cup V(P_2^v))) \setminus  \{u_1,u_2\}$, and suppose without loss of generality that $w \in V(\widehat{P}_1^u) \cap V(P_1^v)$; see Figure~\ref{fig:lem:two-cases}(a) for an illustration. We define a good spindle $\widehat{S}$ in $G$ as follows. The tail and head of $\widehat{S}$ are vertices $u$ and $v$, respectively. The first path of $\widehat{S}$ starts at $u$, follows $\widehat{P}_1^u$ until its first intersection with $P_1^v$ (vertex $w$ in Figure~\ref{fig:lem:two-cases}(a)), which is distinct from $u_1$ by hypothesis, and then follows $P_1^v$ until $v$. The second path of  $\widehat{S}$ starts at $u$, follows $\widehat{P}_2^u$ until its first intersection with $P_2^v$, which may be vertex $u_2$, and then follows $P_2^v$ until $v$. Since $|B|  \geq \ell-1$ and each of $\widehat{P}_1^u$ and $\widehat{P}_2^u$ intersects exactly one of  $P_1^v$ and $P_2^v$, it follows that $\widehat{S}$ is indeed a good spindle. On the other hand, since $|V(\widehat{P}_1^u) \cup V(\widehat{P}_2^u)| = |V((P_1^u) \cup V(P_2^u)|$ and vertex $w$ comes strictly after $u_1$ in $P_1^v$, it follows that the first path of $\widehat{S}$ is strictly shorter than the corresponding path of $S$, while the second one is not longer. Therefore, $|V(\widehat{S})| < |V(S)|$, a contradiction to the choice of $S$.

\begin{figure}[h!]
\begin{center}
\includegraphics[width=.95\textwidth]{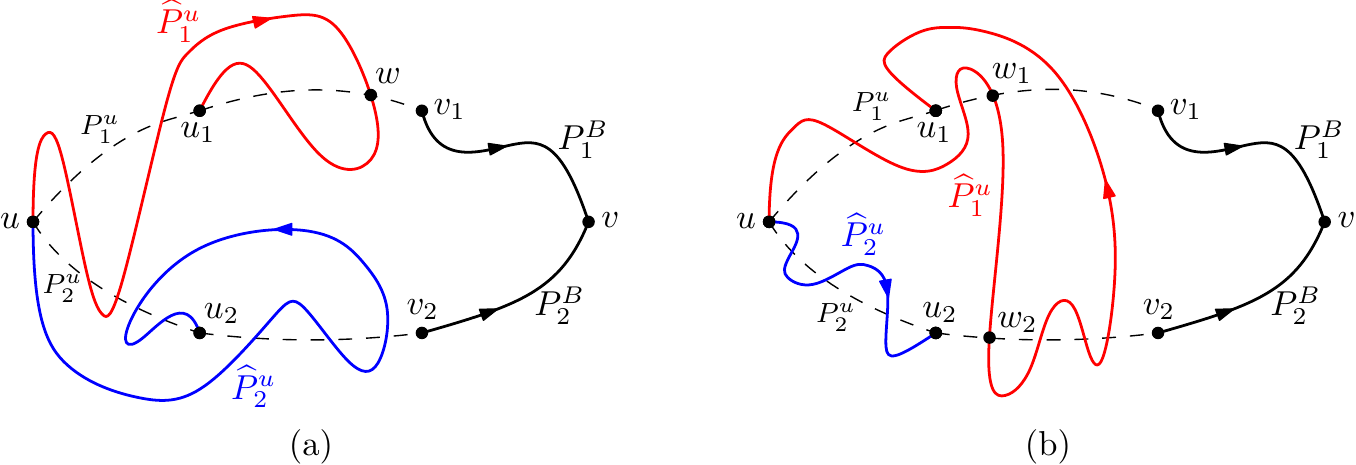}
\end{center}\vspace{-.15cm}
\caption{Illustration of the two cases in the proof of Lemma~\ref{lem:key-lemma}.}
\label{fig:lem:two-cases}
\end{figure}

Suppose now that one of the paths $\widehat{P}_1^u$ and $\widehat{P}_2^u$, say $\widehat{P}_1^u$, intersects both $P_1^v$ and $P_2^v$. Without loss of generality, suppose that, starting from $u$, $\widehat{P}_1^u$ meets $P_1^v$ before than $P_2^v$.  Let $w_1$ and $w_2$ be vertices of $\widehat{P}_1^u$ such that $w_1 \in V(P_1^v)$, $w_2 \in V(P_2^v)$, and there is no vertex of $\widehat{P}_1^u$ between $w_1$ and $w_2$ that belongs to
$V(P_1^v) \cup V(P_2^v)$; see Figure~\ref{fig:lem:two-cases}(b) for an illustration. We define a good spindle $\widehat{S}$ in $G$ as follows. The tail and head of $\widehat{S}$ are vertices $w_1$ and $v$, respectively. The first path of $\widehat{S}$ starts at $w_1$ and follows  $P_1^v$ until $v$. The second path of  $\widehat{S}$ starts at $w_1$, follows $\widehat{P}_1^u$ until $w_2$, and then follows $P_2^v$ until $v$. By the choice of $w_1$ and $w_2$ and since $|B|  \geq \ell-1$, it follows that $\widehat{S}$ is indeed a good spindle.  On the other hand, by construction $|V(\widehat{S})| \leq |V(S)| - |V(\widehat{P}_2^u)| < |V(S)|$, contradicting again the choice of $S$. \end{proof}

\noindent
\textbf{Wrapping up the algorithm}. We finally have all the ingredients to describe our algorithm, which proceeds as follows. First, for every triple of  vertices $u,u_1,u_2 \in V(G)$ and positive integers $\ell_1, \ell_2, q$ with $\ell_1,\ell_2 \leq \ell$ and $q \leq 2 \ell - (\ell_1 + \ell_2)$,
we compute, as explained in Section~\ref{sec:compute-representatives},   a $q$-representative family
$
\widehat{\mathcal{S}}_{u,u_1,u_2}^{\ell_1,\ell_2,q} \subseteq^q_{\text{rep}} \mathcal{S}_{u,u_1,u_2}^{\ell_1,\ell_2}
$
of size $2^{O(\ell)}$ in time $2^{O(\ell)}\cdot n^{O(1)}$. Then the algorithm checks, for each $u,u_1,u_2,v \in V(G)$, integers $\ell_1,\ell_2,q$ with $\min\{\ell_1,\ell_2\} \geq 1$, $\ell_1 + \ell_2 = \ell $, and  $q \leq \ell -1$, and set $S \in \widehat{\mathcal{S}}_{u,u_1,u_2}^{\ell_1,\ell_2,q}$, whether $G$ contains a
$(u_1,v)$-path $P_1^v$ and a $(u_2,v)$-path $P_2^v$ such that $V(P_1^v) \cap V(P_2^v) = \{v\}$ and $S \cap (V(P_1^v) \cup V(P_2^v)) = \{u_1,u_2\}$. Note that the latter check can be easily performed in polynomial time by a flow algorithm~\cite{BJ-Gutin-book}. The correctness of the algorithm follows directly from Lemma~\ref{lem:structural-1} and Lemma~\ref{lem:key-lemma}, and its running time is $2^{O(\ell)}\cdot n^{O(1)}$, as claimed. In order to keep the exposition as simple as possible, we did not focus on optimizing either the constants involved in the algorithm or the degree of the polynomial factor. Nevertheless, explicit small constants can be derived by carefully following the details in Fomin et al.~\cite{FominLPS16}.

\subsubsection{Finding 2-spindles with two specified lengths}
\label{sec:FPT-2}

We now turn to the problem of finding 2-spindles with two specified lengths. Namely, given a digraph $G$ and two  integers $\ell_1, \ell_2$ with $\ell_2 \geq \ell_1 \geq 1$, our objective is to decide whether $G$ contains a subdivision of a $(\ell_1,\ell_2)$-spindle  in time $2^{O(\ell_2)}\cdot n^{O(\ell_1)}$. Note that this problem differs from the one considered in Section~\ref{sec:FPT-1}, as now we specify {\sl both} lengths of the desired spindle, instead of just its total size.  Our approach is similar to the one presented in Section~\ref{sec:FPT-1}, although some more technical ingredients are needed, and we need to look at the problem from a slightly different point of view.

In this section, we say that a subdigraph $S$ of $G$  is a \emph{good spindle} if it is a subdivision of a $(\ell_1,\ell_2)$-spindle. We may again assume that $\max\{\ell_1,\ell_2\} \geq 2$. The following lemma plays a similar role as Lemma~\ref{lem:structural-1}, but now we will exploit the fact that our algorithm can afford to guess the first $\ell_1$ vertices in the ``short'' path. Its proof is also easy to verify. See Figure~\ref{fig:lem:structural-2} for an illustration.


\begin{lemma}\label{lem:structural-2}
A digraph $G$ has a good spindle if and only if there exist vertices $u,u',v$,
a $(u,v)$-path $P_1$ of length at  least $\ell_1$,
a $(u,u')$-path $P_2^u$ on $\ell_2$ vertices,
and a $(u',v)$-path $P^v_2$ such that
$V(P_1) \cap V(P_2^u) = \{u\}$,
 $V(P_1) \cap V(P_2^v) = \{v\}$, and
 $V(P_2^u) \cap V(P_2^v)= \{u'\}$.
 \end{lemma}

 \begin{figure}[h!]
\begin{center}\vspace{-.1cm}
\includegraphics[width=.45\textwidth]{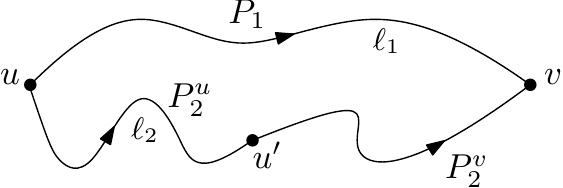}
\end{center}\vspace{-.15cm}
\caption{Illustration of the vertices and paths described in Lemma~\ref{lem:structural-2}.}
\label{fig:lem:structural-2}
\end{figure}

The main difference with respect to Section~\ref{sec:FPT-1} is that now we will only represent the candidates for the first $\ell_2$ vertices of the ``long'' path, denoted by $V(P_2^u)$ in Lemma~\ref{lem:structural-2}.  To this end, we define, similarly to~\cite{FominLPS16}, the following set for every pair of vertices $u,u' \in V(G)$ and positive integer $\ell_2$:
\begin{align*}
\mathcal{P}_{u,u'}^{\ell_2} =  \Big\{ X :  \    S\subseteq V(G), |X| = \ell_2, \text{ and $G[X]$ contains a $(u,u')$-path on $\ell_2$ vertices}\Big\}.
\end{align*}

The above sets are exactly the same as those defined by Fomin et al.~\cite{FominLPS16} to solve the \textsc{Long Directed Cycle} problem. Therefore, we can just apply~\cite[Lemma 5.2]{FominLPS16} and compute,  for every pair of vertices $u,u' \in V(G)$ and positive integers $\ell_2, q$ with $q \leq \ell_1 + \ell_2 \leq 2 \ell_2$,
a $q$-representative family
$$
\widehat{\mathcal{P}}_{u,u'}^{\ell_2,q} \subseteq^q_{\text{rep}}  \mathcal{P}_{u,u'}^{\ell_2}
$$
of size $2^{O(\ell_2)}$ in time $2^{O(\ell_2)}\cdot n^{O(1)}$.

Now we would like to state the equivalent of Lemma~\ref{lem:key-lemma} adapted to the new representative families. However, it turns out that the families $\widehat{\mathcal{P}}_{u,u'}^{\ell_2,q}$ are not yet enough in order to find the desired spindle. To circumvent this cul-de-sac, we use the following trick: we first try to find ``short'' spindles using the color-coding technique of Alon et al.~\cite{AlonYZ95}, and if we do not succeed, we can guarantee that all good spindles have at least one ``long'' path. In this situation, we can prove that the families $\widehat{\mathcal{P}}_{u,u'}^{\ell_2,q}$ are indeed enough to find a good spindle.
More precisely, a good spindle $S$ is said to be \emph{short} if both its paths have at most $2\ell_2$ vertices, and it is said to be \emph{long} otherwise. Note that the following lemma only applies to digraphs without  good short spindles.


\begin{lemma}\label{lem:key-lemma2}
Let $G$ be a digraph containing no good short spindles.  If $G$ contains a good long spindle, then there exist vertices $u,u',v$,
a $(u,v)$-path $P_1$ of length at  least $\ell_1$,
a $(u,u')$-path $\widehat{P}_2^u$ on $\ell_2$ vertices such that  $V(\widehat{P}_2^u) \in \widehat{\mathcal{P}}_{u,u'}^{\ell_2,q}$ with $q = \ell_1 + \ell_2 - 1$,
and a $(u',v)$-path $P^v_2$ such that
$V(P_1) \cap V(\widehat{P}_2^u) = \{u\}$,
 $V(P_1) \cap V(P_2^v) = \{v\}$, and
 $V(\widehat{P}_2^u) \cap V(P_2^v)= \{u'\}$.
\end{lemma}
\begin{proof}
Let $S$ be a good spindle in $G$ with minimum number of vertices, which exists by hypothesis, and let $u$ and $v$ be the tail and the head of $S$, respectively.  Let $P_1$ be the shortest of the two $(u,v)$-paths of $S$, and let $u'$ be the vertex on the other path of $S$ at distance exactly $\ell_2 - 1 $ from $u$. Let $P_2^u$ and $P_2^v$ be the $(u,u')$-path and the $(u',v)$-path in $S$, respectively. Note that $P_2^u \in \mathcal{P}_{u,u'}^{\ell_2}$. Since by hypothesis $S$ is a{\sl long} spindle, it follows that $|V(P_2^v)| >  \ell_2$.

Again, in order to apply the properties of $q$-representative families, we define a vertex set $B \subseteq V(S)$ as follows, crucially using the hypothesis that $S$ is a good long spindle. Namely, $B$ contains the last $\ell_1$ vertices of the path $P_1$ together with the last $\ell_2$ vertices of the path $P_2^v$, including $v$. Note that $|B| = \ell_1 + \ell_2 -1$ and that, since $|V(P_2^v)| >  \ell_2$, we have $V(P_2^u) \cap B = \emptyset$.

Let $q = |B|$. Since $P_2^u \in \mathcal{P}_{u,u'}^{\ell_2}$ and $V(P_2^u) \cap B = \emptyset$, by definition of $q$-representative family there exists a set in
$\widehat{\mathcal{P}}_{u,u'}^{\ell_2,q}$ corresponding to a $(u,u')$-path $\widehat{P}_2^u$ such that $V(\widehat{P}_2^u) \cap B = \emptyset$.  We claim that $V(\widehat{P}_2^u) \cap V(S) \subseteq V(P_2^u)$, which concludes the proof of the lemma. Assume for contradiction that $(V(\widehat{P}_2^u) \cap V(S)) \setminus V(P_2^u) \neq \emptyset$, and we again distinguish two cases.

Suppose first that $\widehat{P}_2^u$ is disjoint from $P_1$, except for vertex $u$. Let $w$ be the first vertex of $\widehat{P}_2^u$ in $V(P_2^v) \setminus \{u'\}$; see Figure~\ref{fig:lem:two-cases-2}(a) for an illustration. We define a good spindle $\widehat{S}$ in $G$ as follows. The tail and head of $\widehat{S}$ are vertices $u$ and $v$, respectively. The first path of $\widehat{S}$ is equal to $P_1$. The second path of $\widehat{S}$ starts at $u$, follows $\widehat{P}_2^u$ until its first intersection with $P_1^v$ (vertex $w$ in Figure~\ref{fig:lem:two-cases-2}(a)), which is distinct from $u_1$ by hypothesis, and then follows $P_1^v$ until $v$. By definition of $B$, it follows that $\widehat{S}$ is a good spindle, and by construction $|V(\widehat{S})| < |V(S)|$, a contradiction to the choice of $S$.

\begin{figure}[h!]
\begin{center}
\includegraphics[width=.95\textwidth]{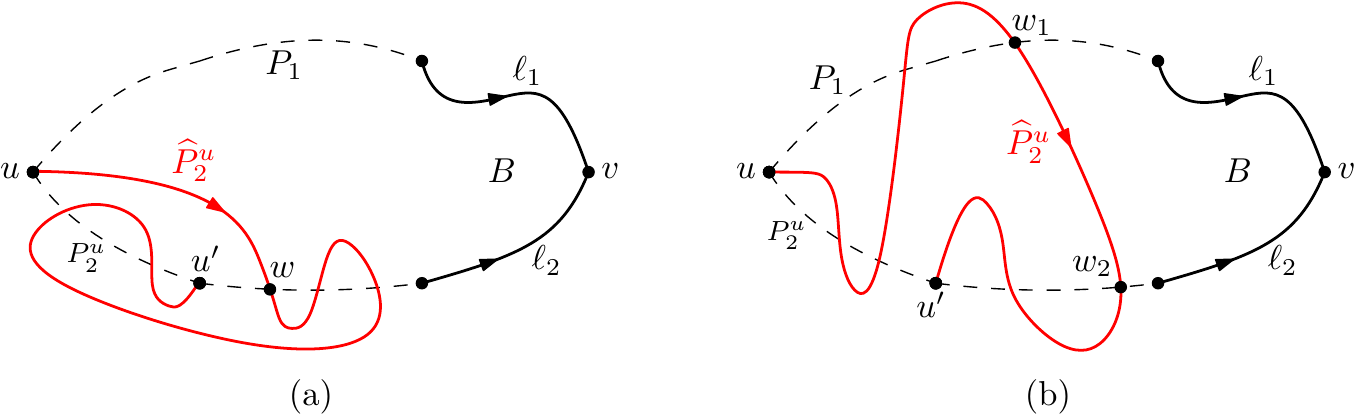}
\end{center}\vspace{-.15cm}
\caption{Illustration of the two cases in the proof of Lemma~\ref{lem:key-lemma2}.}
\label{fig:lem:two-cases-2}
\end{figure}

Suppose now that $\widehat{P}_2^u$ intersects $P_1$. Since $\widehat{P}_2^u$ ends at vertex $u' \notin V(P_1)$, there exist vertices $w_1, w_2$ such that $w_1 \in V(P_1)$, $w_2 \in V(\widehat{P}_2^u) \cup V(P_2^v) $, and there is no vertex of $\widehat{P}_1^u$ between $w_1$ and $w_2$ that belongs to
$V(P_1) \cup V(\widehat{P}_2^u) \cup V(P_2^v)$; see Figure~\ref{fig:lem:two-cases-2}(b) for an illustration. We define a good spindle $\widehat{S}$ in $G$ as follows. The tail and head of $\widehat{S}$ are vertices $w_1$ and $v$, respectively.  The first path of $\widehat{S}$ starts at $w_1$ and follows  $P_1$ until $v$. The second path of  $\widehat{S}$ starts at $w_1$, follows $\widehat{P}_2^u$ until $w_2$, and then follows $\widehat{P}_2^u \cup P_2^v$ until $v$. By the choices of $B$, $w_1$, and $w_2$, it follows that $\widehat{S}$ is a good spindle with $|V(\widehat{S})|  < |V(S)|$, contradicting again the choice of $S$. \end{proof}

\noindent
\textbf{Wrapping up the algorithm.} We start by trying to find good small spindles. Namely, for every pair of integers $\ell_1', \ell_2'$ with $\ell_1 \leq \ell_1' \leq 2\ell_2$ and $\ell_2 \leq \ell_2' \leq 2\ell_2$, we test whether $G$ contains a $(\ell_1', \ell_2')$-spindle as a {\sl subgraph}, by using the color-coding technique of Alon et al.~\cite{AlonYZ95}. Since the treewidth of an undirected spindle is two, this procedure takes time $2^{O(\ell_2)}\cdot n^{O(1)}$.

If we succeed, the algorithm stops. Otherwise, we can guarantee that $G$ does not contain any good short spindle, and  therefore we are in position to apply Lemma~\ref{lem:key-lemma2}.  Before this, we first compute, for every pair of vertices $u,u' \in V(G)$ and positive integers $\ell_2, q$ with $q \leq \ell_1 + \ell_2 \leq 2 \ell_2$, a $q$-representative family
$
\widehat{\mathcal{P}}_{u,u'}^{\ell_2,q} \subseteq^q_{\text{rep}}  \mathcal{P}_{u,u'}^{\ell_2}
$
of size $2^{O(\ell_2)}$ in time $2^{O(\ell_2)}\cdot n^{O(1)}$, using~\cite[Lemma 5.2]{FominLPS16}.

Now, for each path $\widehat{P}_2^u$ such that $V(\widehat{P}_2^u) \in \widehat{\mathcal{P}}_{u,u'}^{\ell_2,q}$, with $q = \ell_1 + \ell_2 - 1$, we proceed as follows. By  Lemma~\ref{lem:structural-2} and Lemma~\ref{lem:key-lemma2}, it is enough to guess a vertex $v \in V(G)$ and check whether
$G$ contains a $(u,v)$-path $P_1$ of length at  least $\ell_1$,
and a $(u',v)$-path $P^v_2$ such that
$V(P_1) \cap V(\widehat{P}_2^u) = \{u\}$,
 $V(P_1) \cap V(P_2^v) = \{v\}$, and
 $V(\widehat{P}_2^u) \cap V(P_2^v)= \{u'\}$. In order to do so, we apply brute force  and we guess the first $\ell_1$ vertices of $P_1$ in time $n^{O(\ell_1)}$. Let these vertices be $u, u_2, \ldots, u_{\ell_1}$. All that remains is to test whether the graph $G \setminus \{u_2,\ldots, u_{\ell_1 -1}\} \setminus (V(\widehat{P}_2^u) \setminus \{u'\})$ contains two internally vertex-disjoint paths from $u_{\ell_1}$ and $u'$ to $u$, which can be done in polynomial time by using a flow algorithm~\cite{BJ-Gutin-book}.
 The correctness of the algorithm follows by the above discussion, and its running time is
$2^{O(\ell_2)}\cdot n^{O(\ell_1)}$, as claimed. Again, we did not focus on optimizing the constants involved in the algorithm.

\subsubsection{Computing the representative families efficiently}
\label{sec:compute-representatives}

We now explain how the representative families used in Sections~\ref{sec:FPT-1} and~\ref{sec:FPT-2} can be efficiently computed, by using the results of Fomin et al.~\cite{FominLPS16}. As discussed in Section~\ref{sec:FPT-2}, the families $\widehat{\mathcal{P}}_{u,u'}^{\ell_2,q}$ are exactly the same as those used by Fomin et al.~\cite{FominLPS16}, so we can directly use~\cite[Lemma 5.2]{FominLPS16} and compute them in time $2^{O(\ell_2)}\cdot n^{O(1)}$. Let us now explain how the results of Fomin et al.~\cite{FominLPS16} can be used to compute efficiently the families $\widehat{\mathcal{S}}_{u,u_1,u_2}^{\ell_1,\ell_2,q}$ used in Section~\ref{sec:FPT-1}. We need the following lemma.

\begin{lemma}[Fomin et al.~\cite{FominLPS16}]\label{lem:union}
Let $M=(E,\mathcal{I})$ be a matroid a $\mathcal{S}$ be a family of subsets of $E$. If
$\mathcal{S}= \mathcal{S}_1 \cup \dots \cup \mathcal{S}_{k}$ and
$\widehat{\mathcal{S}}_i \subseteq^q_{\text{rep}} \mathcal{S}_i$ for $1 \leq i \leq k$, then
$\cup_{i=1}^k \widehat{\mathcal{S}}_i \subseteq^q_{\text{rep}} \mathcal{S}$.
\end{lemma}

The key observation is that the families $\mathcal{S}_{u,u_1,u_2}^{\ell_1,\ell_2}$ can be obtained by combining pairs of elements in the families $\mathcal{P}_{u,u'}^{\ell_2}$. More precisely, for every triple of vertices $u,u_1, u_2$ and positive integers $\ell_1,\ell_2$, it holds that
$$
\mathcal{S}_{u,u_1,u_2}^{\ell_1,\ell_2}  \subseteq \mathcal{P}_{u,u_1}^{\ell_1} \cup \mathcal{P}_{u,u_2}^{\ell_2}.
$$
Note that in the above equation we do not have equality, as some pairs of paths in $\mathcal{P}_{u,u_1}^{\ell_1}$ and $\mathcal{P}_{u,u_2}^{\ell_2}$, respectively, may intersect at other vertices distinct from $u$.

By Lemma~\ref{lem:union}, if $\widehat{\mathcal{P}}_{u,u_1}^{\ell_1,q} \subseteq^q_{\text{rep}} \mathcal{P}_{u,u_1}^{\ell_1}$ and $\widehat{\mathcal{P}}_{u,u_2}^{\ell_2,q} \subseteq^q_{\text{rep}} \mathcal{P}_{u,u_2}^{\ell_2}$, then
$$
\widehat{\mathcal{P}}_{u,u_1}^{\ell_1,q} \cup \widehat{\mathcal{P}}_{u,u_2}^{\ell_2,q}
\subseteq^q_{\text{rep}} \mathcal{P}_{u,u_1}^{\ell_1} \cup
\mathcal{P}_{u,u_2}^{\ell_2}.
$$
To conclude, it just remains to observe that, by the definition of $q$-representative family, it holds that if $M=(E,\mathcal{I})$ is a matroid, $\mathcal{S}$ is a family of subsets of $E$, $\mathcal{S}' \subseteq \mathcal{S}$ and
$\widehat{\mathcal{S}} \subseteq^q_{\text{rep}} \mathcal{S}$, then
$\widehat{\mathcal{S}} \subseteq^q_{\text{rep}} \mathcal{S}'$
as well.

Therefore, for every triple of vertices $u,u_1, u_2$ and positive integers $\ell_1,\ell_2$ with $\ell_1 , \ell_2 \leq \ell$, in order to compute a $q$-representative family for $\mathcal{S}_{u,u_1,u_2}^{\ell_1,\ell_2}$, we can just take the union of $q$-representative families for $\mathcal{P}_{u,u_1}^{\ell_1}$ and $\mathcal{P}_{u,u_2}^{\ell_2}$, and these latter families can be computed in time $2^{O(\ell)}\cdot n^{O(1)}$ by~\cite[Lemma 5.2]{FominLPS16}.

\section{Finding spindles on directed acyclic graphs}
\label{sec:DAGs}

In this section we focus on the case where the input digraph is acyclic.
We start by proving Theorem~\ref{thm:algoDAG}. The proof uses classical dynamic programming along a topological ordering of the vertices of the input acyclic digraph.

\bigskip

\begin{proof-DAG} 
Given an acyclic digraph $G$ and positive integers $k,\ell$, recall that we want to prove that one can decide  in time $O(\ell^k \cdot n^{2k+1})$ whether $G$ has a subdivision of a $(k\times \ell)$-spindle. For this, let $H$ be obtained from the empty digraph by adding, for each vertex $u\in V(G)$, vertices $u^+, u^-$ and an arc $(u^+,u^-)$ between them, and adding arc $(u^-,v^+)$  for each arc $(u,v)\in A(G)$. Note that $H$ is also acyclic, and fix an arbitrary topological ordering of $V(H)$. 

\begin{claim}\label{claim:DP-transf}
There exists a subdivision of a $(k\times\ell)$-spindle in $G$ if and only if there exist $x,y \in V(H)$ and $k$ arc-disjoint $(x,y)$-paths in $H$, each of length at least $2\ell-1$.
\end{claim}
\begin{proof-claim}
On the one hand, each path of a $(k\times \ell)$-spindle gives rise to a path in $H$ of length at least $2\ell-1$, since each internal vertex of a path is split into two (these paths are actually vertex-disjoint). On the other hand, let $P_1,\ldots, P_k$ be arc-disjoint paths between $x,y\in V(H)$, each of length at least $2\ell-1$. Since either $\lvert N^+(z)\rvert = 1$ or $\lvert N^-(z) \rvert=1$ for every $z\in V(H)$, and since $P_1,\ldots,P_k$ are arc-disjoint, we get that $P_1,\ldots,P_k$ are actually internally vertex-disjoint. Now, to obtain the desired $(k\times\ell)$-spindle, it suffices to observe that if $u^+\in V(P_i)\setminus\{y\}$, for some $u\in V(G)$ and some $i\in \{1,\ldots, k\}$, then $u^-\in V(P_i)$.
\end{proof-claim}

We want to decide whether $H$ has the desired paths. For each $x\in V(H)$, we define the table $P_x$ with entries $(e_1,t_1,\ldots,e_k,t_k)$, for each choice of at most $k$ distinct arcs $e_1,\ldots,e_k$ (some of these may not exist, in which case we represent it by `{\sf null}'), and for each choice of $k$ values $t_1,\ldots,t_k$ from the set $\{0,1,\ldots,2\ell - 1\}$. Observe that $P_x$ has size $(\lvert A(H)\rvert +1)^k \cdot (2\ell)^k$, which, since we need to analyze the table of each vertex, gives us the claimed complexity of the algorithm. The meaning of an entry is given below:

\medskip
$P_x(e_1,t_1,\ldots,e_k,t_k) = {\sf true}$ if and only if there exist $k$ arc-disjoint paths $P_1,\ldots,P_k$\\ \indent starting at $x$ and ending at $e_1,\ldots,e_k$ of length at least $t_1,\ldots,t_k$, respectively.
\medskip

\noindent We compute these tables starting at small values of $\sum_{i=1}^k t_i$. Namely, for $t_1=t_2=\ldots=t_k=0$, it holds that
$P_x(e_1,t_1,\ldots,e_k,t_k) = {\sf true} \mbox{ if and only if }\{e_1,\ldots,e_k\}=\emptyset.$

Now, to compute $P_x(e_1,t_1,\ldots,e_k,t_k)$, let $w$ be the greatest vertex in $\{z\in V(H) : (z',z)\in \{e_1,\ldots,e_k\}\}$, and let $w'$ be the greatest vertex in $\{z\in V(H) : (z,w)\in \{e_1,\ldots,e_k\}\}$, according to the chosen topological ordering of $V(H)$. Also, let $e_i=(w',w)$. If $w=x$, then the entry is given above, so suppose otherwise.

 \begin{claim}\label{claim:DP}
$P_x(e_1,t_1,\ldots,e_k,t_k)={\sf true}$ if and only if $P_x(e_1,t_1,\ldots,e,t_i-1,\ldots,e_k,t_k)={\sf true}$, for some arc $e\in A(H) \setminus \{e_1,\ldots,e_k\}$ incoming at $w'$.
\end{claim}
%

\begin{proof-claim}
Suppose first that $P_x(e_1,t_1,\ldots,e_k,t_k)={\sf true}$, and  let $P_1,\ldots,P_k$ be arc-disjoint paths starting at $x$ and ending at $e_1,\ldots,e_k$ of length at least $t_1,\ldots,t_k$, respectively. Let $e$ be the arc preceding $e_i$ in path $P_i$ ($e$ can denote the empty set when $e_i$ is incident to $x$). Then, $P_1,\ldots, P_{i-1},P_i-e_i, P_{i+1},\ldots, P_k$ are arc-disjoint paths ending at $e_1,\ldots,e_{i-1}, e, e_{i+1}, \ldots,e_k$ of length at least $t_1,\ldots,t_{i-1}, t_i-1,t_{i+1},\ldots,t_k$, respectively.

Conversely, let $P_1,\ldots,P_k$ be arc-disjoint paths that certify entry $P_x(e_1,t_1,\ldots,e,t_i-1,\ldots,e_k,t_k)$. If $e_i\notin A(P_j)$, for every $j\in \{1,\ldots,k\}$, then $P_1,\ldots,P_{i-1},P_i+e_i,P_{i+1},\ldots,P_k$ are the desired paths. So suppose that $e_i\in A(P_j)$. If $j=i$, then we get a cycle in $H$, a contradiction. Otherwise, because $e_i\neq e_j$ and $P_j$ ends in $e_j$, we get that there is a path starting in $w$ and ending in $z'$, where $e_j = (z,z')$. This contradicts the choice of $w$.
\end{proof-claim}

By Claim~\ref{claim:DP},  the entry $P_x(e_1,t_1,\ldots,e_k,t_k)$ can be computed by verifying at most $\lvert N^+(w')\rvert$ smaller entries. By Claim~\ref{claim:DP-transf}, the desired spindle exists if and only if there exist $x, y \in V(H)$ and $k$ arcs $e_1, \ldots, e_k$ incoming at $y$ such that $P_x(e_1,2\ell -1, e_2, 2\ell -1, \ldots,e_k,2\ell -1)={\sf true}$. The theorem follows.
\end{proof-DAG}

Motivated by the fact that finding a subdivision of a general digraph $F$ is in \xp parameterized by $|V(F)|$ on acyclic digraphs~\cite{BHM15,Met93}, we now present two hardness results about finding subdivisions of disjoint spindles on acyclic digraphs. The first result holds even for {\sl planar} acyclic digraphs.

\begin{proposition}
\label{prop:disjoint-NPhard} If $F$ is the disjoint union of $(2 \times 1)$-spindles, then deciding whether a planar acyclic digraph contains a subdivision of $F$ is {\sf NP}-complete.
\end{proposition}
\begin{proof}
We reduce from the problem of deciding whether the edges of a tripartite graph can be partitioned into triangles, which is known to be {\sf NP}-complete~\cite{GareyJ79comp}, even restricted to planar tripartite graphs~\cite{VePa96}. Let $G$ be an input planar tripartite (undirected) graph, and let $A \cup B \cup C$ be a tripartition of $V(G)$. We build from $G$ a planar acyclic digraph $G'$ by orienting all edges from $A$ to $B$, from $B$ to $C$, and from $A$ to $C$. It is clear that $E(G)$ admits a partition into triangles if and only if $G'$ contains as a subdivision (in fact, as a subdigraph) the digraph containing $|E(G)|/3$ disjoint copies of a $(2 \times 1)$-spindle.
 \end{proof}


Our next result shows that, for some choices of $F$, finding a subdivision of $F$ is {\sf W}$[1]$-hard on acyclic digraphs.  We just present a sketch of proof, as the reduction  is based on a minor modification of an existing reduction of Slivkins~\cite{Sli10}.

\begin{proposition}
\label{prop:disjoint-Whard} If $F$ is  the disjoint union of a $(k_1 \times 1)$-spindle and a $(k_2 \times 1)$-spindle, then deciding whether an acyclic digraph contains a subdivision of $F$ is {\sf W}$[1]$-hard  parameterized by $k_1 + k_2$.
\end{proposition}
\begin{proof-sketch}
The proof is done by appropriately modifying the reduction for \textsc{Edge-Disjoint Paths} on acyclic digraphs given by Slivkins~\cite{Sli10}, which carries over to the vertex-disjoint version as well. The reduction is from $k$-\textsc{Clique}, and the sets of demands to be satisfied consist just of a multiarc with multiplicity ${k \choose 2}$ and another one with multiplicity $k$ between two given pairs of terminals. The idea is the following: since in our problem we do not have {\sl fixed} terminals, we ``simulate'' them by leaving only four vertices of high degree, so that finding the desired subdivision will only be possible by using the prescribed four vertices as endpoints. To do so, we take the construction of Slivkins~\cite{Sli10} and for each vertex, except for the four prescribed ones, we replace its outgoing (resp. incoming) arcs by an out-arborescence (resp. in-arborescence) of out-degree (resp. in-degree) at most two. Note that this operation may blow up the size of the subdivision, but it does not matter, as the parameter remains the same. By taking $F$ to be the disjoint union of a $({k \choose 2} \times 1)$-spindle and a $(k \times 1)$-spindle, the result follows.
\end{proof-sketch}

It is worth noting that the problem considered in Proposition~\ref{prop:disjoint-Whard} is para-{\sf NP}-hard on general digraphs, as the conditions of~\cite[Theorem 8]{BHM15} are easily seen to be fulfilled.

\section{Conclusions}
\label{sec:further}

We studied the complexity of several problems consisting in finding subdivisions of spindles on digraphs. For a general spindle $F$, we do not know if finding a subdivision of $F$ is \fpt on general digraphs parameterized by $|V(F)|$, although we believe that it is indeed the case. As a partial result, one could try to prove that, for a {\sl fixed} value of $\ell \geq 4$, finding a subdivision of a $(k \times \ell)$-spindle is {\sf FPT}  parameterized by $k$ (the problem is {\sf NP}-hard by Theorem~\ref{thm:dichotomy}).

The above question is open even if the input digraph is acyclic (note that Theorem~\ref{thm:algoDAG} does {\sl not} answer this question), or even if $F$ is a 2-spindle.  Concerning 2-spindles, one may try use the technique we used to prove Theorems~\ref{thm:2blockFPTsum} and~\ref{thm:2blockFPT-2length}, based on representative families in matroids. However, the technique does not seem to be easily applicable when the parameter is the total size of a prescribed 2-spindle. Namely, using the terminology from Section~\ref{sec:FPT-2}, the bottleneck is to find spindles that have one ``short'' and one ``long'' path. On the other hand, generalizing this technique to spindles with more than two paths seems pretty complicated.

Cai and Ye~\cite{CaiY16} recently studied the problem of finding two edge-disjoint paths on undirected graphs with length constraints between specified vertices. These length constraints can be an upper bound, a lower bound, or an equality on the lengths of each of the two desired paths, or no restriction at all,  resulting in nine different problems. Interestingly, out of these nine problems, Cai and Ye~\cite{CaiY16} gave {\sf FPT} algorithms for seven of them, and left open only the following two cases: when there is only one constraint of type `at least', and when both constraints are of type `at least'. Interestingly, this latter problem is closely related to finding a subdivision of a 2-spindle.

In general, very little is known about the complexity of finding subdivisions on digraphs.
Bang-Jensen et al.~\cite{BHM15} conjectured that, considering $|V(F)|$ as a constant, the problem of finding a subdivision of $F$ is either polynomial-time solvable or {\sf NP}-complete. This conjecture is wide open.
Recently, Havet et al.~\cite{HMM17} studied the cases where $|V(F)|= 4$, and  managed to classify all of them up to five exceptions.  Even less is known about the parameterized complexity of the cases that are polynomial-time solvable for fixed $F$, that is, the cases in \xp. In this article we focused on spindles, but there are other potential candidates such as, using the terminology of~\cite{BHM15}, windmills, palms, or antipaths.


\bibliographystyle{abbrv}
\bibliography{Digraph-bib}

\end{document}